\documentclass[12pt]{article}
\usepackage[pdftex]{graphicx}
\usepackage{sidecap}
\usepackage{amsthm}
\usepackage{epsfig,color}

\usepackage{framed}
\newcommand{\comment}[1]{}
\newcommand{\QED}{\mbox{}\hfill \rule{3pt}{8pt}\vspace{10pt}\par}
\usepackage{booktabs} %,times}
\usepackage{algpseudocode,float}
\usepackage[ruled,vlined,linesnumbered]{algorithm2e}
\usepackage[cmex10]{amsmath}
\usepackage{amssymb,mathtools}
\usepackage{color}
\usepackage[table, dvipsnames]{xcolor}
\usepackage{graphicx}
\usepackage{subfig}
\usepackage{wrapfig}
\usepackage{verbatim} %,url}
\usepackage{enumerate}
\usepackage{pifont}
\usepackage{hyperref}

\newtheorem{theorem}{Theorem}[section]
\newtheorem{lemma}[theorem]{Lemma}

\newtheorem{corollary}[theorem]{Corollary}
\newtheorem{definition}{Definition}
\newtheorem{observation}{Observation}

\newcommand{\dis}{{\sc Dispersion}}

\newcommand{{\rh}}{{\widehat r}}
\newcommand{{\Rh}}{{\widehat R}}

\newcommand{\cR}{{\mathcal R}}

\usepackage{mathtools}

\newcommand{\gok}[1]{{\color{red}\bf [Gokarna: #1]}}
\newcommand{\ajay}[1]{{\color{green}\bf [Ajay: #1]}}

\newboolean{short}
\setboolean{short}{false}
\newcommand{\shortOnly}[1]{\ifthenelse{\boolean{short}}{#1}{}}
\newcommand{\onlyShort}[1]{\ifthenelse{\boolean{short}}{#1}{}}
\newcommand{\longOnly}[1]{\ifthenelse{\boolean{short}}{}{#1}}
\newcommand{\onlyLong}[1]{\ifthenelse{\boolean{short}}{}{#1}}
\newcommand{\shortLong}[2]{\ifthenelse{\boolean{short}}{#2}{#1}}
\newcommand{\longShort}[2]{\ifthenelse{\boolean{short}}{#2}{#1}} 
\onlyShort{

\usepackage{enumitem}
\setitemize{noitemsep,topsep=0pt,parsep=0pt,partopsep=0pt}
\setenumerate{noitemsep,topsep=0pt,parsep=0pt,partopsep=0pt}
}
    
\begin{document}

%
% The "title" command has an optional parameter, allowing the author to define a "short title" to be used in page headers.

\title{Efficient Dispersion of Mobile Robots on Arbitrary Graphs and Grids}

\author{Ajay D. Kshemkalyani \thanks{University of Illinois at Chicago, llinois 60607, USA. \hbox{E-mail}:~{\tt ajay@uic.edu}.} \and Anisur Rahaman Molla \thanks{Indian Statistical Institute, Kolkata 700108, India.  \hbox{E-mail}:~{\tt molla@isical.ac.in}. Research supported in part by DST Inspire Faculty research grant DST/INSPIRE/04/2015/002801.} \and Gokarna Sharma
\thanks{Kent State University, Ohio 44240, USA.  \hbox{E-mail}:~{\tt sharma@cs.kent.edu}.}}

\iffalse
%
% The "author" command and its associated commands are used to define the authors and their affiliations.
% Of note is the shared affiliation of the first two authors, and the "authornote" and "authornotemark" commands
% used to denote shared contribution to the research.
\author{Ajay D. Kshemkalyani}
\affiliation{%
  \institution{University of Illinois at Chicago}
  %\streetaddress{P.O. Box 1212}
  \state{Illinois 60607}
  \country{USA}}
  \email{ajay@uic.edu}

\author{Anisur Rahaman Molla}
\affiliation{%
  \institution{Indian Statistical Institute}
  \city{Kolkata 700108}
  \country{India}}
\email{molla@isical.ac.in}

\author{Gokarna Sharma}
\affiliation{%
  \institution{Kent State University}
  \city{Ohio 44240}
  \country{USA}}
\email{sharma@cs.kent.edu}

\fi
%
% By default, the full list of authors will be used in the page headers. Often, this list is too long, and will overlap
% other information printed in the page headers. This command allows the author to define a more concise list
% of authors' names for this purpose.
%\renewcommand{\shortauthors}{A. D. Kshemkalyani, A. R. Molla, and G. Sharma}

% The abstract is a short summary of the work to be presented in the article.

\date{}

\maketitle \thispagestyle{empty}

\maketitle

\begin{abstract}
The mobile robot dispersion problem on graphs asks $k\leq n$ robots placed initially arbitrarily on the nodes of an $n$-node anonymous graph to reposition autonomously to reach a configuration in which each robot is on a distinct node of the graph. This problem is of significant interest due to its relationship to other fundamental robot coordination problems, such as exploration, scattering, load balancing, and relocation of self-driven electric cars (robots) to recharge stations (nodes). 
In this paper, we provide two novel deterministic algorithms for dispersion, one for  arbitrary graphs and another for grid graphs, in a synchronous setting where all robots perform their actions in every time step.  Our algorithm for arbitrary graphs has $O(\min(m,k\Delta) \cdot \log k)$ steps runtime using $O(\log n)$ bits of memory at each robot, where $m$ is the number of edges and $\Delta$ is the maximum degree of the graph. 
This is an exponential improvement over the $O(mk)$ steps best previously known algorithm. % using the same memory of $\Theta(\log (\max(k,\Delta)))$ bits at each robot. % with runtime $O(mk)$ steps. 
In particular, the runtime of our algorithm is optimal (up to a $O(\log k)$ factor) in constant-degree arbitrary graphs. 
Our algorithm for grid graphs has $O(\min(k,\sqrt{n}))$ steps runtime using $\Theta(\log k)$ bits at each robot. This is the first algorithm for dispersion in grid graphs. Moreover, this algorithm is optimal for both memory and time when $k=\Omega(n)$. 
\end{abstract}

\iffalse
\begin{comment}
\begin{CCSXML}
<ccs2012>
 <concept>
  <concept_id>10010520.10010553.10010562</concept_id>
  <concept_desc>Computer systems organization~Embedded systems</concept_desc>
  <concept_significance>500</concept_significance>
 </concept>
 <concept>
  <concept_id>10010520.10010575.10010755</concept_id>
  <concept_desc>Computer systems organization~Redundancy</concept_desc>
  <concept_significance>300</concept_significance>
 </concept>
 <concept>
  <concept_id>10010520.10010553.10010554</concept_id>
  <concept_desc>Computer systems organization~Robotics</concept_desc>
  <concept_significance>100</concept_significance>
 </concept>
 <concept>
  <concept_id>10003033.10003083.10003095</concept_id>
  <concept_desc>Networks~Network reliability</concept_desc>
  <concept_significance>100</concept_significance>
 </concept>
</ccs2012>
\end{CCSXML}

\ccsdesc[500]{Mathematics of computing~Graph algorithms}
%\ccsdesc[500]{Theory of computation~Distributed algorithms}
\ccsdesc[500]{Computing methodologies~Distributed algorithms}
\ccsdesc[500]{Computer systems organization~Robotics}
\end{comment}
\fi
%
% Keywords. The author(s) should pick words that accurately describe the work being
% presented. Separate the keywords with commas.

%\maketitle

\section{Introduction}
%{\bf TODO: may be knowledge of $m$ is not needed. Clarify and change line 5 of algorithm and other places. Ajay, we will discuss this}
%\gok{Use LIPIcs formating}
%\gok{Discussion on 02/05/2019: (i) special case of $O(k^2\log k )$bound when $k<\Delta$, (ii) relate more with load balancing, (iii) finish grid write up, (iv) async extensions, and (v) make notations consistent throughout the paper and revise Section I-III. New Comment: may be figure.}
%\onlyShort{\vspace{-3mm}}
The dispersion of autonomous mobile robots to spread them out evenly
in a region is a problem of significant interest in distributed robotics, e.g., see \cite{Hsiang:2003,HsiangABFM02}. Recently, this problem has been formulated by Augustine and Moses Jr.~\cite{Augustine2018} in the context of graphs. % (instead of a region).  %Augustine and Moses Jr.~\cite{Augustine2018} 
They defined  the problem as follows:  Given any arbitrary initial configuration of $k\leq n$ robots positioned on the nodes of an $n$-node graph, the robots reposition autonomously to reach a configuration where each robot is positioned on a distinct node of the graph (which we call the {\dis} problem).  
This problem has many practical applications, for example, in relocating self-driven electric cars (robots) to recharge stations (nodes), assuming that the cars have smart devices to communicate with each other to find a free/empty charging station \cite{Augustine2018,Kshemkalyani}.
This problem is also important due to its relationship to many other well-studied autonomous robot coordination problems,
%in distributed robotics, 
such as  exploration, scattering, load balancing, covering, and self-deployment~\cite{Augustine2018,Kshemkalyani}. One of the key aspects of mobile-robot research is to understand how to use the resource-limited robots to accomplish some large task in a distributed manner \cite{Flocchini2012,flocchini2019}. In this paper, we study the trade-off between memory requirement of robots and the time to solve {\dis} on graphs.        
%on special graphs (such as lines, grids, rings, and trees) as well as on arbitrary graphs \cite{}.

Augustine and Moses Jr.~\cite{Augustine2018} studied {\dis} 
assuming $k=n$. 
They proved a memory lower bound of $\Omega(\log n)$ bits at each robot and a time lower bound of $\Omega(D)$ ($\Omega(n)$ in arbitrary graphs) for any deterministic algorithm in any graph, where $D$ is the diameter of the graph.  
%They also showed that any deterministic algorithm
%needs $\Omega(D)$ time in any graph and  %. They also showed that this bound 
%it can be improved to $\Omega(n)$ for arbitrary graphs, where $D$ is the diameter of the graph.  
%Regarding upper bounds, 
They then %Augustine and Moses Jr.~\cite{Augustine2018} 
provided deterministic algorithms using $O(\log n)$ bits at each robot to solve {\dis} on lines, rings, and trees in $O(n)$ time. For arbitrary graphs, they provided two algorithms, one using $O(\log n)$ bits at each robot with $O(mn)$ time and another using $O(n\log n)$ bits at each robot with $O(m)$ time, where $m$ is the number of edges in the graph.
Recently, Kshemkalyani and Ali \cite{Kshemkalyani} provided an $\Omega(k)$ time lower bound for arbitrary graphs for $k\leq n$. % on the robot-only memory model.  
They then provided three deterministic algorithms for {\dis} in arbitrary graphs: (i) The first algorithm using $O(k\log \Delta)$ bits at each robot with $O(m)$ time, (ii) The second algorithm using $O(D\log \Delta)$ bits at each robot with $O(\Delta^D)$ time, and (iii) The third algorithm using $O(\log(\max(k,\Delta)))$ bits at each robot with $O(mk)$ time, where $\Delta$ is the maximum degree of the graph. Randomized algorithms are presented in  \cite{tamc19} to solve {\dis} where the random bits are mainly used to reduce the memory requirement at each robot.   

%\onlyShort{In this paper, we  
%provide two new deterministic algorithms for solving {\dis}, one for arbitrary graphs and another for grid graphs.}
%\onlyLong{
In this paper, we  
provide two new deterministic algorithms for solving {\dis}, one for arbitrary graphs and another for grid graphs. Our algorithm for arbitrary graphs improves exponentially on the runtime of the best previously known algorithm; %  using same amount of memory at each robot; %ning time using memory $\Theta(\log(\max\{\Delta,k\}))$ bits at each robot. 
see Table \ref{table:comparision}. % for the comparison of the results. 
Our algorithm for grid graphs is the first algorithm for {\dis} in grid graphs %(no previous solution for {\dis} on a grid) 
and it achieves bounds that are both memory and time optimal for $k=\Omega(n)$; see Table \ref{table:comparision-grid}. 
%This is the first time such optimal bounds are achieved for {\dis} in non-trivial graphs beyond line and cycle graphs. 
%This is the first optimal algorithm w.r.t. both memory and time for {\dis} in non-trivial graphs (i.e., beyond lines and cycles). 

\vspace{1mm}
\noindent{\bf Overview of the Model and Results.} We consider the same model as in Augustine and Moses Jr.~\cite{Augustine2018} and Kshemkalyani and Ali \cite{Kshemkalyani} where a system of $k\leq n$ robots are operating on an $n$-node anonymous graph $G$.  
The robots are {\em distinguishable}, i.e., they have unique IDs in the range $[1,k]$. % and $k$ is known to robots. 
The robots have no visibility; but they can communicate with each other only when they are at the same node of $G$. % and this is done via passing messages. 
The graph $G$ is assumed to be connected and undirected. % graph with $m$ edges, diameter $D$, and maximum degree $\Delta$. 
The nodes of $G$ 
are indistinguishable ($G$ is anonymous) but the ports (leading to incident
edges) at each node have unique labels from $[1,\delta]$, where $\delta$ is the degree of that node.
It is assumed that 
the robots know $m, n, \Delta, k$\footnote{In fact, it is enough to know only $m, \Delta$ and $k$ to accomplish the results. Without robots knowing $m$, Theorem \ref{theorem:0} achieves {\dis} in $O(k\Delta \cdot \log k)$ time with $O(\log(\max(k,\Delta)))$ bits memory at each robot, which is better in terms of memory of $O(\log n)$ bits in Theorem \ref{theorem:0} but not the time $O(\min(m,k\Delta)\cdot \log k)$  when $m<k\Delta$.}. Similar assumptions are made in the previous work in {\dis} \cite{Augustine2018}. 
%When a robot moves from its current position (node in $G$), it reaches its neighbor node in $G$. 
The nodes of $G$ do not have memory and the robots have memory.
{\em Synchronous} setting is considered as in \cite{Augustine2018} where all robots are activated in a round and they perform their operations simultaneously in synchronized rounds. %(We will discuss an asynchronous setting in Section \ref{section:extensions}.) 
Runtime is measured in rounds (or steps).
%In this paper,
We establish the following theorem % for {\dis} 
in an arbitrary graph.

%\onlyShort{\vspace{-2mm}}
\begin{theorem}
\label{theorem:0}
Given any initial configuration of $k\leq n$ mobile robots in an arbitrary, anonymous %, connected, and undirected 
$n$-node graph $G$ having %$n$ nodes, 
$m$ edges and maximum degree $\Delta$, {\dis} can be solved in $O(\min(m,k\Delta)\cdot \log k)$ time with $O(\log n)$ bits at each robot. %\onlyShort{({\bf cf.  Section~\ref{section:algorithm0}})}
\end{theorem}

Theorem \ref{theorem:0}  improves exponentially over the $O(mk)$ time best previously known algorithm \cite{Kshemkalyani} % with the same memory requirement 
(see Table \ref{table:comparision}). %For $k=n$, Theorem \ref{theorem:0} improves the $O(mn)$ time of Augustine and Moses Jr.~\cite{Augustine2018} % 
%to $O(m \log n)$ time
%from $O(mn)$ to $O(\min\{m,\Delta k\})$ 
%with the same memory requirement of $O(\log n)$ bits at each robot (note $\Delta\leq n-1$). 
Notice that, when $\Delta \leq k$, the runtime depends only on $k$, i.e.,   $O(k^2 \log k)$. %\anis{Cor 17 holds for $\Delta \leq k$, but here we saying $k\leq \Delta$ (?)}\gok{sorry, changed to $\Delta \leq k$, which is correct, not $k\leq \Delta$} 
For constant-degree arbitrary graphs (i.e., when $\Delta=O(1)$), the dispersion time becomes near-optimal -- only a $O(\log k)$ factor away from the time lower bound $\Omega(k)$.  
%We also prove that $\Omega(\log (\max(k,\Delta)))$ bits  at each robot is the memory lower bound. 
%
%the a new memory lower bound (Theorem \ref{theorem:lower}) in this model which shows that, for $k<n$, $\Omega(\log(\max\{\Delta,k\}))$-bits are required at each robot. 
%That is, Theorem \ref{theorem:0} is  memory-optimal. %  for {\dis} in arbitrary graphs. % Also, this algorithm is time-optimal for bounded-degree graphs ($\Delta=O(1)$) since there is a time lower bound of $\Omega(k)$ due to \cite{Kshemkalyani}.

We establish the following theorem %for {\dis} 
in a grid graph.

%\onlyShort{\vspace{-2mm}}
\begin{theorem}
\label{theorem:1}
Given any initial configuration of $k\leq n$ mobile robots in an anonymous %, connected, and undirected 
$\sqrt{n}\times \sqrt{n}$-node grid graph $G$, 
{\dis} can be solved in $O(\min(k,\sqrt{n}))$ time with $\Theta(\log k)$ bits at each robot. %\onlyShort{({\bf cf. Section~\ref{section:grid}})}
\end{theorem}

The time bound in Theorem \ref{theorem:1} is optimal  when $k=\Omega(n)$ since a time lower bound of $\Omega(D)$ holds for {\dis} in any graph and the diameter of a square grid graph is $D=\Omega(\sqrt{n})$. This is the first result for {\dis} in grid graphs. % (no previous solution for grid graphs). 
Furthermore, this is the first memory and time optimal algorithm for {\dis} in graphs beyond trivial graphs such as lines and cycles. %, there is no existing algorithm for {\dis} that is optimal for both memory and time.  
%Using the 
Theorem \ref{theorem:1} extends for {\dis} in a rectangular grid graph $G$.

\begin{table}[!t]
{\footnotesize
\centering
\begin{tabular}{lll}
\toprule
{\bf Algorithm} & {\bf Memory per robot (in bits)}   & {\bf Time (in rounds)}   \\
% & {\bf (in bits)}   & {\bf (in rounds)} \\
\toprule
Lower bound & $\Omega(\log(\max(k,\Delta)))$     & $\Omega(k)$   \\
%Augustine and Moses Jr.~
\hline
First algo. of \cite{Augustine2018}\footnotemark & $O(\log n)$     & $O(mn)$  \\
\hline
%Augustine and Moses Jr.~
Second algo. of \cite{Augustine2018} & $O(n\log n)$   & $O(m)$   \\
\hline
%Kshemkalyani and Ali 
First algo. of \cite{Kshemkalyani} & $O(k\log \Delta)$     & $O(m)$  \\
\hline
%Kshemkalyani and Ali 
Second algo. of \cite{Kshemkalyani} & $O(D \log \Delta)$     & $O(\Delta^D)$  \\
\hline
%Kshemkalyani and Ali 
Third algo. of \cite{Kshemkalyani} & $O(\log(\max(k,\Delta)))$     & $O(mk)$  \\
\hline
 %\rowcolor{lightgray}
{\bf Theorem \ref{theorem:0}} %(this paper)} 
& $O(\log n)$      & $O(\min(m,k\Delta)\cdot \log k)$  \\
\bottomrule
\end{tabular}
%\vspace{-2mm}
\caption{The results on {\dis} for $k\leq n$ robots on $n$-node arbitrary graphs with $m$ edges, $D$ diameter, and $\Delta$ maximum degree. $^2$The results in \cite{Augustine2018} are only for $k=n$. % and hence the results are given for $k=n$.
}  
\label{table:comparision}
}
%\vspace{-8mm}
\end{table}

\begin{table}[!t]
{\footnotesize
\centering
\begin{tabular}{lll}
\toprule
{\bf Algorithm} & {\bf Memory per robot} & {\bf Time (in rounds)}\\
 & {\bf (in bits)}   &  \\ %{\bf (in rounds)}   \\
% & {\bf (in bits)}   & {\bf (in rounds)} \\
\toprule
Lower bound & $\Omega(\log k)$    & $\Omega(\sqrt{k})$   \\
%Augustine and Moses Jr.~
\hline
Applying first algo. of \cite{Kshemkalyani} & $O(k)$     & $O(n)$  \\
\hline
%Kshemkalyani and Ali 
Applying second algo. of \cite{Kshemkalyani} & $O(D)=O(\sqrt{n})$     & $O(4^D)=O(4^{\sqrt{n}})$  \\
\hline
%Kshemkalyani and Ali 
Applying third algo. of \cite{Kshemkalyani} & $O(\log k)$     & $O(nk)$  \\
\hline
%Kshemkalyani and Ali 
Applying Theorem \ref{theorem:0} & $O(\log n)$     & $O(k \log k)$  \\
\hline
 %\rowcolor{lightgray}
{\bf Theorem \ref{theorem:1}} %(this paper)} 
& $O(\log k)$      & $O(\min(k,\sqrt{n}))$  \\
\bottomrule
\end{tabular}
%\vspace{-2mm}
\caption{The results on {\dis} for $k\leq n$ robots on $n$-node grid graphs (for grids, $\Delta=4$).  
\label{table:comparision-grid}
}}
\vspace{-8mm}
\end{table}

%\onlyLong{
\vspace{1mm}
\noindent{\bf Challenges and Techniques.}
%\gok{this needs modification}
The well-known {\em Depth First Search} (DFS) traversal approach \cite{Cormen:2009} was used in the previous papers to solve  {\dis} \cite{Augustine2018,Kshemkalyani}. 
If all $k$ robots are positioned initially on a single node of $G$, then the DFS traversal finishes in $\min(4m-2n+2,\, k\Delta)$ rounds solving {\dis}.  If $k$ robots are initially on $k$ different nodes of $G$, then {\dis} is solved by doing nothing. 
However, if not all of them are on a single node initially, then the robots on nodes with multiple robots need to reposition (except one) to reach to free nodes and settle. The natural approach is to run DFS traversals in parallel to minimize time.

The challenge arises when two or more DFS traversals meet before all robots settle. When this happens, the robots that have not settled yet need to find free nodes. For this, they may need to re-traverse the already traversed part of the graph by the DFS traversal. Care is needed here otherwise they may re-traverse sequentially and the total time for the DFS traversal increases by a factor of $k$ to $\min(4m-2n+2,k\Delta) \cdot k$ rounds, in the worst-case. This is in fact the case in the previous algorithms of \cite{Augustine2018,Kshemkalyani}.
We design a smarter way to synchronize the parallel DFS traversals so that the total time increases only by a factor of $\log k$ to  $\min(4m-2n+2,k\Delta)\cdot \log k$ rounds, in the worst-case. This approach is a non-trivial extension and requires overcoming many challenges on synchronizing the parallel DFS traversals efficiently.
%\onlyLong{
%For constant-degree arbitrary graphs, the time bound becomes $O(k \log k)$, which is optimal within a factor of $O(\log k)$ (there is a time lower bound of $\Omega(k)$ for arbitrary graphs).   

For grid graphs, applying the DFS traversal approach developed above gives $O(k \log k)$ time (note $\Delta=4$ in a grid). However, in grid, time lower bound is $\Omega(D)=\Omega(\sqrt{n})$ for $k=\Omega(n)$. Therefore, we develop a key technique specific to grid graphs that achieves $O(\min(k,\sqrt{n}))$  time for any $k\leq n$. The grid approach crucially uses the idea of repositioning first all robots to the boundary nodes of $G$, then collect them to a boundary corner node of $G$, and finally distribute them  to the nodes of $G$ leaving one robot on each node.

\vspace{1mm}
\noindent{\bf Related Work.}
\label{section:related}
%\onlyLong{
%We build upon the previous work on {\dis} \cite{Augustine2018,Kshemkalyani}.
%Their results on arbitrary graphs are summarized in Table \ref{table:comparision}. There was no prior study of {\dis} in grid graphs. Table \ref{table:comparision-grid} summarizes the results obtained for grid graphs applying the results of \cite{Augustine2018,Kshemkalyani} and Theorem \ref{theorem:0}. 
% (Table \ref{table:comparision} provides a comparison of our results to theirs). 
%
%}
%
One problem closely related to {\dis} is the graph exploration by mobile robots. The exploration problem has been heavily studied in the literature for specific as well as arbitrary graphs, %under different settings, 
e.g., \cite{Bampas:2009,Cohen:2008,Dereniowski:2015,Fraigniaud:2005,MencPU17}. It was shown that a robot can explore an anonymous graph using $\Theta(D\log \Delta)$-bits memory; the runtime of the algorithm is $O(\Delta^{D+1})$ \cite{Fraigniaud:2005}. In the model where graph nodes also have memory, % in addition to robots, 
Cohen {\it et al.} \cite{Cohen:2008} gave two algorithms: The first algorithm uses $O(1)$-bits at the robot and 2 bits at each node, and the second algorithm uses $O(\log \Delta)$ bits at the robot and 1 bit at each node. The runtime of both algorithms is $O(m)$ with preprocessing time of $O(mD)$. The trade-off between exploration time and number of robots is studied in \cite{MencPU17}. %The Rotor-Router-based algorithms are also studied with respect to both upper and lower bounds \cite{Bampas:2009}. 
The collective exploration by a team of robots is studied in \cite{FraigniaudGKP06} for trees. % giving both upper and lower bounds under different settings. 
Another problem related to {\dis} is the scattering of $k$ robots in graphs. This problem has been studied for rings \cite{ElorB11,Shibata:2016} and grids \cite{Barriere2009}. %under different assumptions. 
Recently, Poudel and Sharma \cite{Poudel18} provided a $\Theta(\sqrt{n})$-time algorithm for uniform scattering in a grid \cite{Das16}. % under synchrony. 
Furthermore, {\dis} is  related to the load balancing problem, where a given
load at the nodes  has to be (re-)distributed among several processors (nodes). This problem has been studied quite heavily in graphs, % through diffusion and dimension-exchange based approaches, 
e.g., \cite{Cybenko:1989,Subramanian:1994}.
We refer readers to %these two excellent books 
\cite{Flocchini2012,flocchini2019} for other recent developments in these topics.

%\onlyLong{
\vspace{1mm}
\noindent{\bf Paper Organization.}  We discuss %related work in Section \ref{section:related} and the 
details of the model and some preliminaries in Section \ref{section:model}. 
We discuss the DFS traversal of a graph in Section \ref{section:basic}. % which will be used heavily in Section \ref{section:algorithm0}. 
We present an algorithm for arbitrary graphs in Section \ref{section:algorithm0}. 
We present an algorithm for grid graphs in Section \ref{section:grid}.
Finally, we conclude in Section \ref{section:conclusion} with a short discussion. % on possible future work.
\onlyShort{Some details and several proofs are omitted from the paper due to space constraints, and are available in the full version \cite{arxiv-full} in arXiv.} % to the full version attached in Appendix.}

%\onlyShort{Due to space constraints, many details and proofs are deferred to the full version attached in Appendix. }
% with a short discussion.

%\onlyShort{\vspace{-3mm}}
\section{Model Details and Preliminaries}
\label{section:model}
%\onlyShort{\vspace{-3mm}}
\noindent{\bf Graph.} We consider the same graph model as in \cite{Augustine2018,Kshemkalyani}. Let $G=(V,E)$ be an $n$-node $m$-edge graph, i.e., $|V|=n$ and $|E|=m$. $G$ is assumed to be connected, unweighted, and undirected. 
%Edges can be though of as {\em bridges} between two nodes, where each node has its own port to denote the bridge. 
%, where $V=\{v_{1}, v_{2},\ldots, v_n\}$ denotes the node sets and  $E\subseteq V\times V$ denotes the edge sets between any two nodes in $V$. Each edge $(v_i,v_j),i\neq j,$ represents a line connecting any two nodes $v_{i}$ and $v_{j}$ of $V$. %We denote by n the  number of nodes in the graph. 
$G$ is {\em anonymous}, i.e., nodes do not have identifiers but, at any node, its incident edges are uniquely identified by a {\em label} (aka port number) in the range $[1,\delta]$, where $\delta$ is the {\em degree} of that node. 
The {\em maximum degree} of $G$ is $\Delta$, which is the maximum among the degree $\delta$ of the nodes in $G$.
%We call the label of an edge at any node the {\em port number}. % at that nose    
We assume that there is no correlation between two port numbers of an edge. 
%From a different perspective, edges can be though of as {\em bridges} between two nodes, where each node has its own port to denote the bridge. 
Any number of robots are allowed to move along an edge at any time. % (unlimited edge bandwidth). %We will discuss in Section \ref{section:extensions} how this can be removed.  
%The edges are unweighted and undirected. 
%A path $p(v_i,v_j)$ is a finite sequence of edges which connects the nodes in $G$, %$v_{i}$, $v_{2}$, \ldots, $v_{j}$ 
%starting at $v_i$ and ending at $v_j$. The degree of a node is $\delta$, and 
The graph nodes do not have memory, i.e., they are not able to store any information. 

%For a square grid graph $G$, the $4\sqrt{n}-4$ nodes on 2 boundary rows and 2 boundary columns of $G$ are called {\em boundary nodes} and the 4 corner nodes of $G$ on boundary rows (or columns) are called {\em boundary corner nodes}.   

\vspace{1mm}
\noindent{\bf Robots.} We also consider the same robot model as in \cite{Augustine2018,Kshemkalyani}. Let $\cR=\{r_{1}, r_{2},\ldots,r_{k}\}$ be a set of $k\leq n$ robots residing on the nodes of $G$. For simplicity, we sometime use $i$ to denote robot $r_i$. No robot can reside on the edges of $G$, but one or more robots can occupy the same node of $G$.  
%In the initial configuration $C_{init}$, we assume that robots in $\cR$ can be in one or more nodes of $G$ but in the final configuration $C_{final}$ there must be exactly one robot on different nodes of $G$. 
%We denote by $v_{i}$ the node on which  $r_{i}$ currently resides, and $e(r_i,r_j)$ (or $e(v_i,v_j)$) the edge
%connecting two nodes $v_{i}$ and $v_{j}$ of $G$. Note that this is used for description purposes in the algorithm; the nodes in $G$ do not have IDs. 
%
%
Each robot has a unique $\lceil \log k\rceil$-bit ID taken from $[1,k]$. %We assume that robots know $k$ but do not know other parameters $m$, $\Delta$, and current round of the algorithm (more later). 
Robot has no visibility and hence a robot can only communicate with other robots present on the same node. % (except in a model where a robot can pass messages to the neighbors of the node that it is currently positioned). 
%The communication is done via passing messages. 
Following \cite{Augustine2018,Kshemkalyani}, it is assumed that 
when a robot moves from node $u$ to node $v$ in $G$, it is aware of the port of $u$ it used to leave $u$ and the port of $v$ it used to enter $v$. 
%
%We do not restrict the time of local computation of the robots and communication that can take place among themselves. The only guarantee is that all this happens in a cycle and we measure time with respect to the number of cycles.   
%
Furthermore, it is assumed that each robot is equipped with memory to store  information, which may also be read and modified by other robots on the same node. Each robot is assumed to know parameters $m,n,\Delta,k$. Such assumptions are also made in the previous work on {\dis} \cite{Augustine2018} .

\vspace{1mm}
\noindent{\bf Time Cycle.}
At any time a robot $r_i\in \cR$ could be active or inactive. When a robot $r_i$ becomes active, it performs
the ``Communicate-Compute-Move'' (CCM) cycle as follows. 
\begin{itemize}
\item 
{\em Communicate:} For each robot $r_j\in \cR$ that is at node $v_i$ where $r_i$ is,  $r_i$ can observe the memory of $r_j$.  
Robot~$r_i$ can also observe its own memory. 
\item {\em Compute:} $r_i$ may perform an arbitrary computation
using the information observed during the ``communicate'' portion of
that cycle. This includes determination of a (possibly) port to 
use to exit $v_i$ and the information to store in 
%the whiteboard memory at $v_i$ and 
the robot $r_j$ that is at $v_i$.
\item
{\em Move:} At the end of the cycle, $r_i$ writes new information (if any) in the memory of $r_j$ at $v_i$,  and exits $v_i$ using the computed port to reach to a neighbor of $v_i$. 
\end{itemize}

\begin{comment}
%\vspace{1mm}
\noindent{\bf Robot Activation.}
In the synchronous setting, 
every robot is active in every CCM cycle. In the {\em semi-synchronous} setting, at least one robot is active, and over an infinite number of CCM cycles, every robot is active infinitely often. In the {\em asynchronous} setting, there is no common notion of time and no assumption is made on the number and frequency of CCM cycles in which a robot can be active. The only guarantee is that each robot is active infinitely often. %Complying with the {\asynch} setting, we assume that a robot %``wakes up'' and 
%performs
%its {\em Look} phase at an instant of time. 
%We also assume that %during the
%{\em Move} phase 
%a robot moves at some (not necessarily constant)
%speed without stopping or changing direction until it reaches its destination ({\em monotonic} movements).
%
%
\end{comment}

\vspace{1mm}
\noindent{\bf Time and Memory Complexity.}
We consider the synchronous setting where 
every robot is active in every CCM cycle and they perform the cycle in a synchrony.
%For the synchronous setting, 
Therefore, time is measured in {\em rounds} or {\em steps} (a cycle is a round or step).
%Since a robot in the semi-synchronous and asynchronous settings could stay inactive for
%an indeterminate time, we bound a robot's
%inactivity introducing the idea of an epoch. % to measure time.
%An {\em epoch} is the smallest interval of time within which each
%robot is guaranteed to be active at least once \cite{Poudel18}.  
%
Another important parameter is memory. Memory comes from a single source -- the number of bits stored at each robot. % -- for solving {\dis}. % (graph nodes do not have memory). 
%The memory complexity is the number of bits that each robot stores for solving {\dis}. %to facilitate local computation in each CCM cycle.  
%
%

\vspace{1mm}
\noindent{\bf Mobile Robot Dispersion.} The {\dis} problem can be formally defined as follows.
\begin{definition} [{\dis}]
Given any $n$-node anonymous graph $G=(V,E)$ having $k\leq n$ mobile robots positioned initially arbitrarily on the nodes of $G$, the robots reposition autonomously to reach a configuration where each robot is on a distinct node of $G$.
\end{definition}

The goal is to solve {\dis} optimizing two performance metrics: 
(i) {\bf Time} -- the number of rounds (steps), and
(ii) {\bf Memory} -- the number of bits stored at each robot.

\begin{table}[!t]

{\footnotesize
\centering
\begin{tabular}{l|l}
\toprule
{\bf Symbol} & {\bf Description}  \\
\toprule
%$C_{init}$ & Initial configuration of robots position on $V$\\
%\hline
%$N(v)$ &  The set of robots on a node $v\in G$ in $C_{init}$\\
%\hline
%$P(v)$& The set of ports at node $v\in G$;  $|P(v)|=\delta_v$, where $\delta_v$ is the degree of $v$\\
%\hline
%$r_{smallest,v_i}$ & The smallest ID robot among the robots in the set $N(v_i)$\\
%\hline
%$r_{largest,v_i}$& The largest ID robot among the robots in the set $N(v_i)$\\
%\hline
%$port_{smallest,v_i}$ & The smallest ID port  among the ports in the set $P(v_i)$\\ 
%\hline
%$port_{largest,v_i}$ & The largest ID port  among the ports in the set $P(v_i)$ \\
%\hline
%$T(r_{smallest,v_i})$ & The DFS tree labeled by the robot $r_{smallest,v_i}$\\
%\hline
%\hline
$round$ & The counter that indicates the current round. \\
& Initially, $round\leftarrow 0$\\
\hline
$pass$ & The counter that indicates the current pass. Initially, $pass\leftarrow 0$\\
%\hline
%$phase$& A flag that indicates 1 (forward phase) or 2 (backtrack phase). Initially, $phase\leftarrow 1$\\ 
\hline
$parent$ & The port from which robot entered a node in forward phase.\\
& Initially, $parent\leftarrow 0$\\ 
\hline
%$P^{noparent}(v)$ & $P(v)\backslash \{parent\}$\\
%\hline
%$P^{eligible}(v)$ & $P^{noparent}(v)\backslash\{0,1,\ldots,child\}$. Initially, $child\leftarrow null$\\
%\hline
$child$& The smallest port (except $parent$ port) that was not taken yet.\\ & Initially, $child\leftarrow 0$\\
\hline
$port\_entered$ & The port from which robot entered a node. \\
& Initially, $port\_entered\leftarrow -1$\\ 
\hline
%$P^{noparent}(v)$ & $P(v)\backslash \{parent\}$\\
%\hline
%$P^{eligible}(v)$ & $P^{noparent}(v)\backslash\{0,1,\ldots,child\}$. Initially, $child\leftarrow null$\\
%\hline
$port\_exited$& The port from which robot exited a node. \\
& Initially, $port\_exited\leftarrow -1$\\ 
\hline
$treelabel$& The label of a DFS tree. %It is set to the smallest ID unsettled robot in $N(v)$. % at node $v$ in $C_{init}$. 
Initially, $treelabel\leftarrow \top$\\
\hline
$settled$& A boolean flag that stores either  0 (false) or 1 (true). \\
& Initially, $settled \leftarrow 0$\\ 

\hline
$mult$& The number of robots at a node at the start of Stage 2. \\
& Initially, $mult\leftarrow 1$\\
\hline
$home$& The lowest ID unsettled robot at a node at the start of Stage 2 % (Section \ref{section:algorithm0}) %(Algorithm \ref{algo1}) 
\\ 
& sets this to  the ID of the settled robot at that node. \\
&  Initially, $home\leftarrow \top$\\
\bottomrule
\end{tabular}
\caption{Description of the variables used in Sections \ref{section:basic}, \ref{section:algorithm0}, and \ref{section:grid}. These variables are maintained by each robot and may be read/updated by other robots (at the same node).} % during execution.  %\ajay{Gokarna: This table should should be broken into 2 parts: one of non-variables, other of variables (treelabel, parent, child, phase, settled, mult, home)}\gok{OK, good idea}
\label{table:notations}
}

\onlyShort{\vspace{-6mm}}
\end{table}

\begin{comment}
\begin{algorithm*}[!h]
{\small
\tcp{$DFS(k)$ is for any robot $r\in \cR$ that is currently at some node $v\in G$. In $C_{init}$, all $k$ robots are at a single node in $G$. $N(v)$ denotes the set of robots of $\cR$ at $v$.}
\If{$r$ is alone at $v$}{ 
 $r.settled\leftarrow 1$; 
}
\Else
{\If{$r_i$ is the maximum ID robot among the robot in the set $N(v)$} 
{}
}
\caption{Algorithm $DFS(k)$.}
\label{fds(k)}
\end{algorithm*}
}
\end{comment}

%\onlyShort{\vspace{-3mm}}
\section{DFS traversal of a Graph} % (Algorithm~$DFS(k)$)}
\label{section:basic}
%\onlyShort{\vspace{-3mm}}
Consider an $n$-node arbitrary graph $G$ as defined in Section \ref{section:model}. Let $C_{init}$ be the initial configuration of $k\leq n$ robots positioned on a single node, say $v$, of $G$. 
Let the robots on $v$ be represented as $N(v)=\{r_1,\ldots, r_k\}$, where $r_i$ is the robot with ID $i$.  
We describe here a DFS traversal algorithm, $DFS(k)$, that disperses all the robots on the set $N(v)$ to the $k$ nodes of $G$ guaranteeing exactly one robot on each node. 
%$DFS(k)$ is analogous to the DFS traversal algorithm of Augustine and Moses Jr.~\cite{Augustine2018} %designed for arbitrary graphs 
%when all robots are on a single node of $G$ in $C_{init}$. We provide details here since it 
$DFS(k)$ 
will be heavily used in Section \ref{section:algorithm0}.

Each robot $r_i$ stores in its memory four variables $r_i.parent$ (initially assigned $0$), $r_i.child$ (initially assigned $0$), $r_i.treelabel$ (initally assigned $\top$), and $r_i.settled$ (initially assigned 0).  %The graph nodes do not have memory and hence store nothing while executing $DFS(k)$. %In $C_{init}$, all these variables are assigned $null$ for $r_i$.  % at every robot.
%\ajay{$r_i.treelabel$ should be initialized to $\top$, not null. $\top$ is highest value, like +infinity. This is useful for robots that are alone, in Sec 5}
%
\begin{comment}
We denote by 
\begin{itemize}
\item $N_i(v_i)$,  the number of robots in a node $v_i\in G$;
\item $P_i(v_i)$, the port numbers of $v_i\in G$;  $|P_i(v_i)|=\delta_i$, where $\delta_i$ is the degree of $v_i$;
\item $r_{smallest,v_i}$, the smallest ID robot among the robots in the set $N_i(v_i)$;
\item $r_{largest,v_i}$, the largest ID robot among the robots in the set $N_i(v_i)$;
\item $port_{smallest,v_i}$, the smallest ID port  among the ports in the set $P_i(v_i)$; 
\item $port_{largest,v_i}$, the largest ID port  among the ports in the set $P_i(v_i)$; 
\item $T(r_{smallest,v_i})$ be the DFS tree labeled by the robot $r_{smallest,v_i}$;
\item $parent$, the port in the set $P_i(v_i)$ at node $v_i\in G$ from which the robot(s) entered $v_i$. Initially, $parent\leftarrow null$; 
\item $child$, the port in the set $P_i(v_i)\backslash \{parent\}$ at node $v_i\in G$ from which the robot(s) exit $v_i$, except $parent$ port. Initially, $child\leftarrow null$; 
\item $treelabel$, the smallest ID robot  among $N_i(v_i)$ robots at node $v_i\in G$ in $C_{init}$.
\item $settled$, a boolean flag that stores either  0 (false) or 1 (true). 
\item $phase$, a flag that indicates 0 (forward) or 1 (backtrack). 
\end{itemize}
\end{comment}
%
%
 %
%We first introduce some notations. 
$DFS(k)$ executes in two phases, $forward$ and $backtrack$ \cite{Cormen:2009}. %Variable $r_i.phase$ is assigned  $1$ ($2$) to denote the forward (backtrack) phase of $DFS(k)$. %Initially, $r_i.phase$ is assigned 1 denoting the forward phase. %
Variable $r_i.treelabel$ stores the ID of the smallest ID robot. Variable $r_i.parent$ stores the port from which $r_i$ entered the node where it is currently positioned in the forward phase.  
%That is, if $p_x$ is the port at some node $x\in G$ from which $r_i$ entered $x$ in the forward phase, then $r_i.parent\leftarrow p_x$. 
Variable $r_i.child$ stores the smallest port of the node  it is currently positioned at that has not been taken yet (while entering/exiting the node). 
%Specifically,
Let $P(x)$ be the set of ports at any node $x\in G$.
%Let $P^{taken}(x)\subseteq P(x)$ be the ports that $r_i$ has already taken in the forward phase. 
%Let $P^{eligible}(x)\leftarrow P(x)\backslash \{P^{taken}(x)\cup p_x\}$. % be the ports at $x$ after removing the ports in the set $P^{taken}(x)$ and the port $p_x$.
%Then, $r_i.child\leftarrow \min(P^{eligible}(x))$.
%Note here that to compute $\min(P^{eligible}(x))$, only the information on variables $r_i.child$ and $r_i.parent$ is enough for $r_i$. That is, $\min(P^{eligible}(x))$ is port $r_i.child+1$ if $r_i.child+1\neq r_i.parent$, otherwise, it is $r_i.child+2$ if $r_i.child+2\leq \delta_x-1$ (the ports at any node $x$ are numbered consecutively from 0 to $\delta_x-1$).  

%Let $p_x$ be the port at $x$ from which $r_i$ entered $$  
\begin{comment}
$P^{noparent}(x)\leftarrow P(x) \backslash \{r_i.parent\}$ be the set of ports at $x$ except $r_i.parent$. 
Let $P^{eligible}(x)\leftarrow P^{noparent}(x) \backslash \{0,1,\ldots,r_i.child\}$ be the ports in the set $P^{noparent}(x)$ after removing the ports with label $\{0,1,\ldots,r_i.child\}$. 
Then, $r_i.child\leftarrow \min(P^{eligible}(x))$.    
\end{comment}

We are now ready to describe $DFS(k)$. 
%$DFS(k)$ proceeds in rounds as follows. 
In round 1,
the maximum ID robot $r_k$ writes  $r_k.treelabel \leftarrow 1$ (the ID of the smallest robot in $N(v)$, which is 1),  $r_k.child\leftarrow 1$ (the smallest port at $v$ among $P(v)$), %\min(P^{eligible}(v))$ (note that in round 1, $P^{eligible}(v)=P(v)$ since $r_k.parent=null$ as $DFS(k)$ started from $v$), 
and $r_k.settled\leftarrow 1$. %, where $P(v)$ are the port numbers at node $v$. %; note that $port_{smallest,v_i}$ denotes the smallest port  among the ports in the set $P_i(v_i)$ in its memory. %Furthermore, $r_{largest,v_i}$ writes $parent \leftarrow null$.
%At the end of round 1, 
The robots $N(v)\backslash\{r_k\}$ exit $v$ following port $r_k.child$; $r_k$ stays (settles) at $v$. 
%The robots in $v$ (except $k$) exit $v$ following port $child$. 
%
In the beginning of round 2, the robots $N(w)=N(v)\backslash\{r_k\}$ reach a neighbor node $w$ of $v$.
Suppose the robots entered $w$ using port $p_w\in P(w)$.
As $w$ is free, robot $r_{k-1}\in N(w)$ writes $r_{k-1}.parent\leftarrow p_w$, $r_{k-1}.treelabel \leftarrow 1$ (the ID of the smallest robot in $N(w)$), and $r_{k-1}.settled\leftarrow 1$.
If $r_{k-1}.child\leq\delta_w$, $r_{k-1}$ writes
$r_{k-1}.child\leftarrow r_{k-1}.child+1$ if port $r_{k-1}.child+1\neq p_w$ and $r_{k-1}.child+1\leq \delta_w$, otherwise $r_{k-1}.child\leftarrow r_{k-1}.child+2$. % if $r_{k-1}.child+2\leq \delta_v$.  %$\min(P^{eligible}(w))$.
The robots $N(w)\backslash\{r_{k-1}\}$ decide to continue DFS in forward or backtrack phase as described below.
%$r_{k-1}$ either keeps the value 1 (forward) in  $r_{k-1}.phase$  or assign 2 (backtrack) as described below. 

\begin{itemize}
\item ({\bf forward phase}) if ($p_w=r_{k-1}.parent$ or $p_w=$ old value of $r_{k-1}.child$) and (there is (at least) a port at $w$ that has not been taken yet).
%\item ({\bf forward phase}) Keeps $r_{k-1}.phase=1$ if there is (at least) a port at $w$ that has not been taken yet. % (this is determined from whether new value in $r_{k-1}.child$ can be assigned). % $r_{k-1}.child\neq null$). %   and $r_{k-1}.child<\max(\{P(w)\backslash \{p_w\}\}$. 
%At the end of round 2, t
The robots $N(w)\backslash\{r_{k-1}\}$ exit $w$ through port $r_{k-1}.child$.  
%This prompts the robots to continue executing Phase 1 (if the robots were executing Phase 1 previously) or transition to Phase 1 from Phase 2 (if the robots were executing Phase 2 previously). 

\item ({\bf backtrack phase}) if ($p_w=r_{k-1}.parent$ or $p_w=$ old value of $r_{k-1}.child$) and (all the ports of $w$ have been taken already).
%\item ({\bf backtrack phase}) Writes $r_{k-1}.phase\leftarrow 2$ if all the ports of $w$ have been taken already. % At the end of round 2, 
The robots $N(w)\backslash\{r_{k-1}\}$ exit $w$ through port $r_{k-1}.parent$.% (the port $p_w$ used to enter $w$). 
\end{itemize}

Assume that in round 2, the robots decide to proceed in forward phase. 
In the beginning of round 3, 
%in the forward phase, 
$N(u)=N(w)\backslash\{r_{k-1}\}$ robots reach some other node $u$ (neighbor of $w$) of $G$. The robot $r_{k-2}$ stays at $u$ writing necessary information in its variables. In the forward phase in round 3, the robots $N(u)\backslash\{r_{k-2}\}$ exit $u$ through port $r_{k-2}.child$.
However, in the backtrack phase in round 3, $r_{k-2}$ stays at $u$ and robots $N(u)\backslash\{r_{k-2}\}$ exit $u$ through port $r_{k-2}.parent$. This takes robots $N(u)\backslash\{r_{k-2}\}$ back to node $w$ along $r_{k-1}.child$. Since $r_{k-1}$ is already at $w$, $r_{k-1}$ updates $r_{k-1}.child$ with the next port to take.
%Depending on whether $r_i.child$ can be set to new value or not, 
Depending on whether $r_i.child\leq \delta_w$ or not,
the robots $\{r_1,\ldots,r_{k-3}\}$ exit $w$ using either $r_{k-1}.child$ (forward phase) or $r_{k-1}.parent$ (backtrack phase). 

There is another condition, denoting the onset of a cycle, under which choosing backtrack phase is in order. When the robots enter $x$ through $p_x$ and robot $r$ is settled at $x$, 
\begin{itemize}
\item ({\bf backtrack phase}) if ($p_x\neq r.parent$ and $p_x\neq$ old value of $r.child$). The robots exit $x$ through port $p_x$ and no variables of $r$ are altered.
\end{itemize}
This process then continues for $DFS(k)$ until at some node $y\in G$, $N(y)=\{r_1\}$. The robot $r_1$ then stays at $y$ and $DFS(k)$ finishes. % writing necessary information on $r_1.parent$, $r_1.treelabel$, and $r_1.settled$. 

\begin{lemma}
Algorithm $DFS(k)$ correctly solves {\dis} for  $k\leq n$ robots initially positioned on a single node of a $n$-node arbitrary graph $G$  in $ \min(4m-2n+2,  k\Delta)$ rounds using $O(\log(\max(k,\Delta)))$ bits at each robot. 
\end{lemma}

\onlyLong{
\begin{proof}
We first show that {\dis} is achieved by $DFS(k)$. Because every robot starts at the same node and follows the same path as other not-yet-settled robots until it is assigned to a node, $DFS(k)$ resembles the DFS traversal of an anonymous port-numbered graph \cite{Augustine2018} with all robots starting from the same node. 
%At each free node visited by $DFS(k)$ the maximum ID robot among the unsettled robots is settled. 
Therefore, $DFS(k)$ visits $k$ different nodes where each robot is settled. 

%\ajay{Gokarna: I think (Delta, k) suffices for algo DFS(k). You need (2 Delta k) for the second stage of the algo in Section 5.}\gok{good point, I will change the proof} 

We now prove time and memory bounds. In $k \Delta$ rounds, $DFS(k)$ visits at least $k$ different nodes of $G$. If $4m-2n+2<  k \Delta$, $DFS(k)$ visits all $n$ nodes of $G$. Therefore, it is clear that the runtime of $DFS(k)$ is $ \min(4m-2n+2, k\Delta)$ rounds. Regarding memory, variable $treelabel$  takes $O(\log k)$ bits, $settled$ takes $O(1)$ bits, and $parent$ and $child$ take $O(\log \Delta)$ bits.
The $k$ robots can be distinguished through $O(\log k)$ bits since their IDs are in the range $[1,k]$.
Thus, each robot requires $O(\log(\max(k,\Delta)))$ bits.  
\end{proof}
}

\section{Algorithm for Arbitrary Graphs} % (Theorem \ref{theorem:0})}
\label{section:algorithm0}
%\onlyShort{\vspace{-3mm}}
We present and analyze an algorithm, {\em Graph\_Disperse(k)}, that solves {\dis} of $k\leq n$ robots on an arbitrary $n$-node graph in $O(\min(m,k\Delta )\cdot \log k)$ time with $O(\log n)$ bits of memory at each robot. 
%\onlyLong{
This algorithm exponentially improves the $O(mk)$ time of the best previously known algorithm \cite{Kshemkalyani} for arbitrary graphs % using the memory of $O(\log (\max(k,\Delta)))$ bits at each robot 
(Table \ref{table:comparision}). % with $O(\log n)$-bits at each robot. 
%The algorithm is time-optimal for bounded-degree graphs.
%We will also show that the robot's memory % used in this algorithm 
%($O(\log (\max(k,\Delta)))$ bits at each robot) is optimal (cf. Theorem \ref{theorem:lower}). % for any $k\leq n$. 

\subsection{High Level Overview of %$Graph\_Disperse(k)$}
the Algorithm}

Algorithm {\em Graph\_Disperse(k)} runs in passes and each pass is divided into two stages. Each pass runs for $O(\min(m,k\Delta ))$ rounds and there will be total $O(\log k)$ passes until {\dis} is solved. Within a pass, each of the two stages runs for $O(\min(m,k\Delta))$ rounds. The algorithm uses $O(\log n)$ bits memory at each robot.
To be able to run passes and stages in the algorithm, we assume following \cite{Augustine2018} that robots know $n, m,k,$ and $\Delta$. At their core, each of the two stages uses a modified version of the DFS traversal by robots (Algorithm $DFS(k)$) described in Section~\ref{section:basic}.

\onlyLong{
\begin{figure*}[!t] %{r}{2.4in}
%\vspace{-5mm}
\centering 
\includegraphics[height=2.7in]{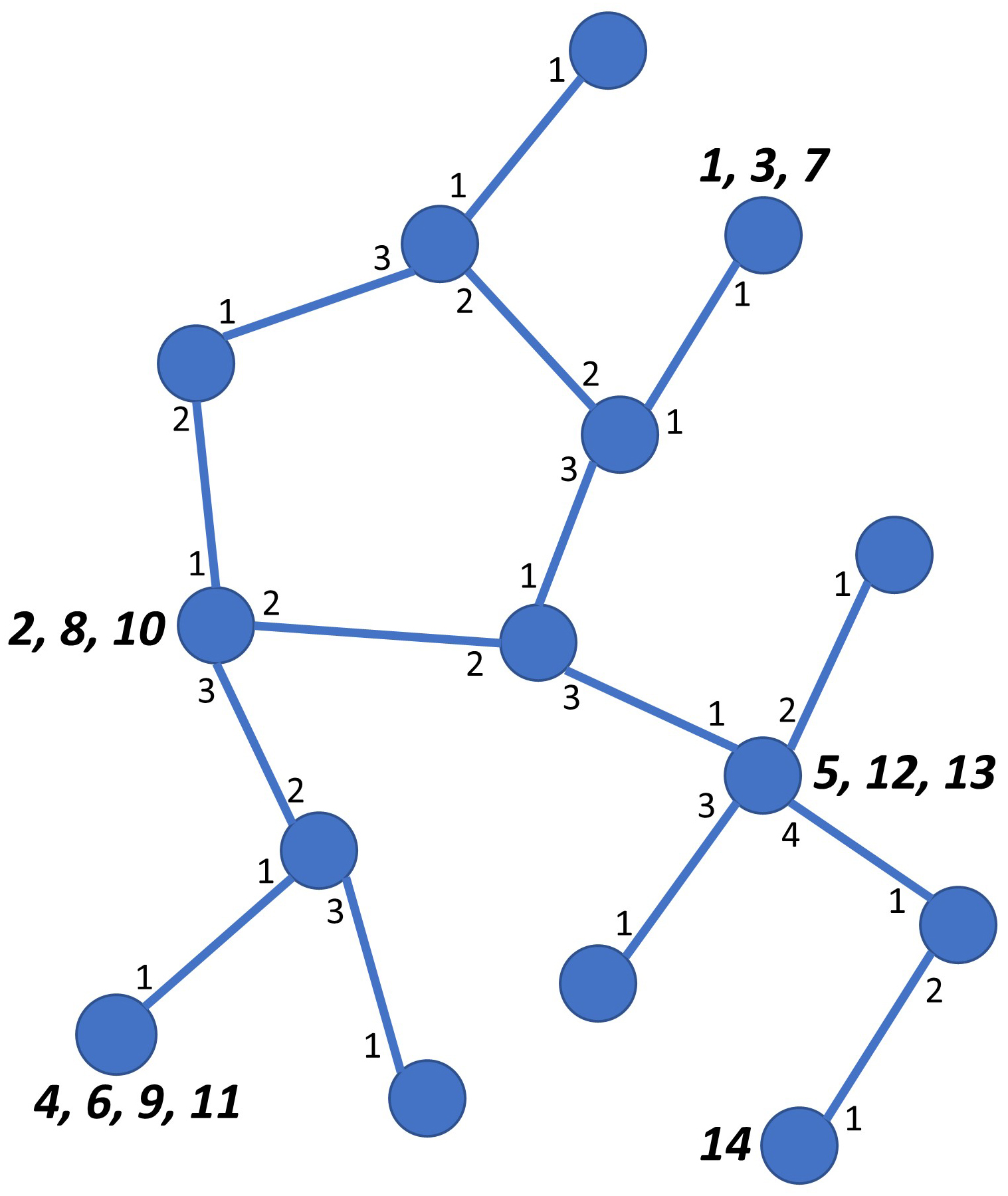}
\hspace{10pt}
\includegraphics[height=2.7in]{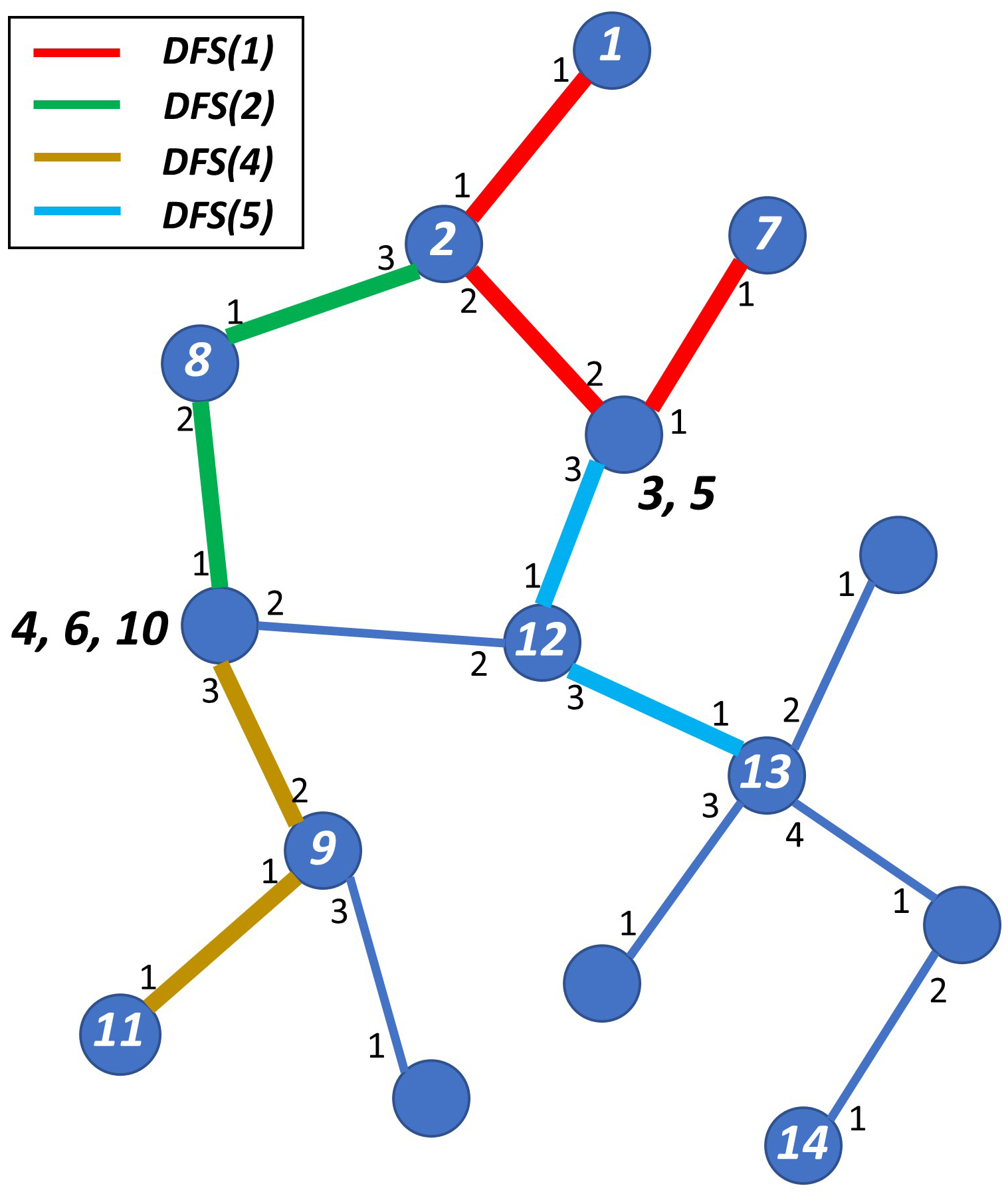}\\
\vspace{2mm}
\includegraphics[height=2.7in]{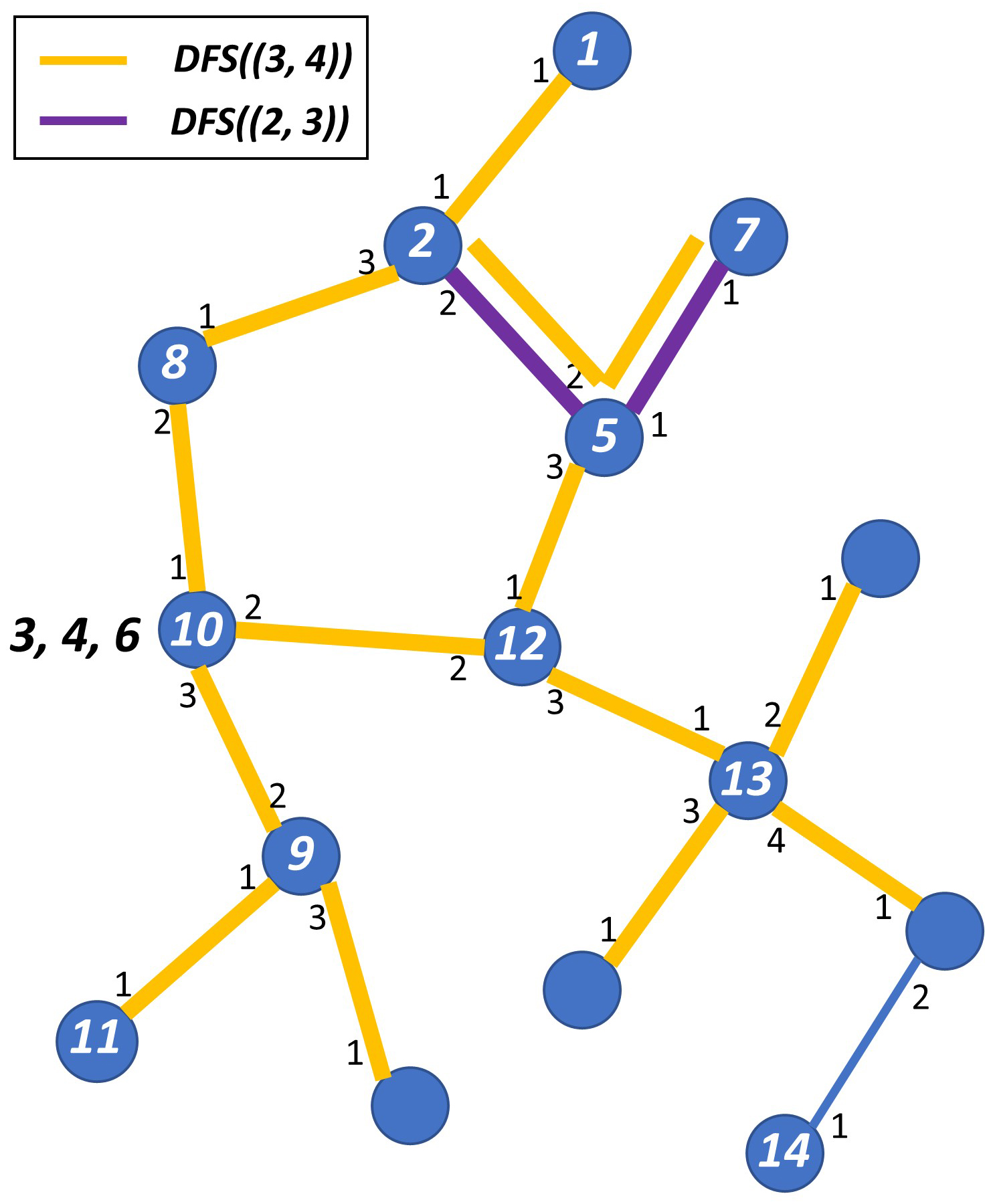}
\hspace{10pt}
\includegraphics[height=2.7in]{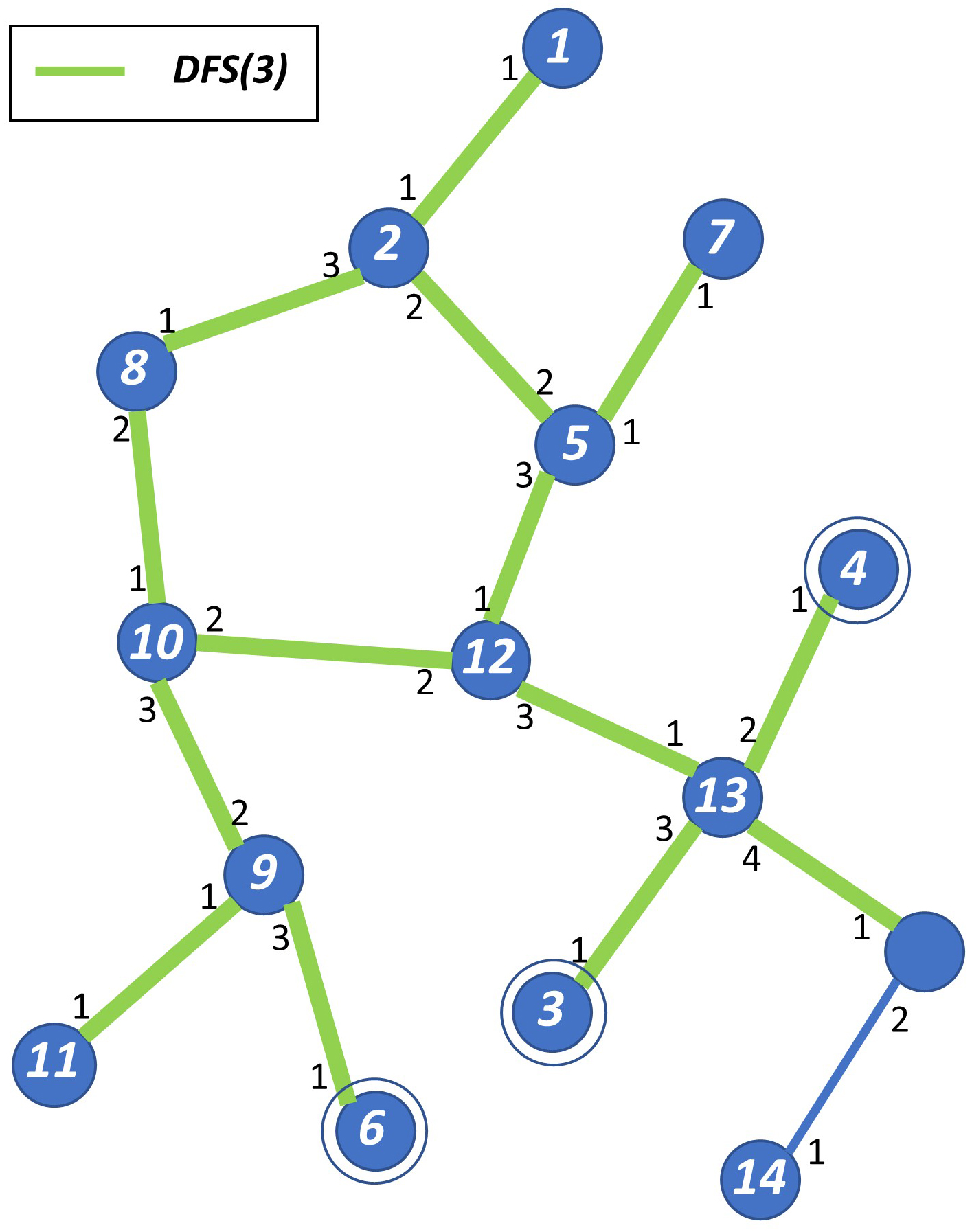}
\vspace{-0mm}
\caption{An illustration of the two stages in a pass of Algorithm \ref{algo1}  for $k=14$ robots in an $15$-node graph $G$. ({\bf top left}) shows $C_{init}$ with one or more robots at 5 nodes of $G$; the rest of the nodes of $G$ are empty. ({\bf top right}) shows the configuration after Stage 1 finishes for $DFS(.)$ started by 4 nodes with multiple robots on them; the respective DFS trees formed are shown through colored edges (the same colored edges belong to the same DFS tree). A single robot (14) at a node settles there.   ({\bf bottom left}) shows the configuration after Stage 2 finishes for $DFS((.,.))$ started by two nodes with more than one robot (see top right) on them when Stage 1 finishes. The robots 3,4,6 are collected at the node of $G$ where robot 10 is settled since $DFS((3,4))$ started from there has higher lexico-priority than $DFS((2,3))$ started from the node of $G$ where $5$ is settled.    ({\bf bottom right}) shows the configuration after Stage 1 of the next pass in which all $k$ robot settle on $k$ different nodes of $G$. There is only one DFS tree $DFS(3)$ started from the node of $G$ (where 10 is settled and all robots are collected in Stage 2) that traverses $G$ until all $3,4,6$ are settled reaching the empty nodes of $G$. The nodes of $G$ where they are settled are also shown inside a circle.}
\vspace{-0mm}
\label{fig:algo1}
\end{figure*}
}

At the start of stage 1, there may be multiple nodes, each with more than one robot\onlyLong{ (top left of Fig.~\ref{fig:algo1})}. 
The (unsettled) robots at each such node begin a DFS in parallel, each such DFS instance akin to $DFS(k)$ described in Section~\ref{section:basic}. Each such concurrently initiated DFS induces a DFS tree where the $treelabel$ of the robots that settle is common, and the same as the ID of the robot with the smallest ID in the group. Unlike $DFS(k)$, here a DFS traversal may reach a node where there is a settled robot belonging to another (concurrently initiated) DFS instance. As the settled robot cannot track variables ($treelabel$, $parent$, $child$) for the multiple DFS trees owing to its limited memory, it tracks only one DFS tree instance and the other DFS instance(s) is/are stopped. Thus, some DFS instances may not run to completion and some of their robots may not be settled by the end of stage 1. Thus, groups of stopped robots exist at different nodes at the end of stage 1\onlyLong{ (top right of Fig.~\ref{fig:algo1})}.

In stage 2, all the groups of stopped robots at different nodes in the same connected component of nodes with settled
robots are gathered together into one group at a single node in that connected component\onlyLong{ (bottom left of Fig.~\ref{fig:algo1})}. Since stopped robots in a group do not know whether there are other groups of stopped robots, and if so, how many and where, one robot from each such group initiates a DFS traversal of its connected component of nodes with settled robots, to gather all the stopped robots at its starting node. The challenge is that due to such parallel initiations of DFS traversals, robots may be in the process of movement and gathering in different parts of the connected component of settled nodes. The algorithm ensures that despite the unknown number of concurrent initiations of the DFS traversals for gathering, all stopped robots in a connected component of settled robots get collected at a single node in that component at the end of stage 2. 
Our algorithm has the property that the number of nodes with such gathered (unsettled) robots in the entire graph at the end of stage 2 is at most half the number of nodes with more than one robot at the start of stage 1 (of the same pass). This implies the sufficiency of $\log k$ passes, each comprised of these two stages, to collect all graph-wide unsettled robots at one node. In the first stage of the last pass, {\dis} is achieved\onlyLong{ (bottom right of Fig.~\ref{fig:algo1})}.

%\onlyShort{\vspace{-5mm}}
\subsection{Detailed Description of the Algorithm}
%\onlyShort{\vspace{-3mm}}
The pseudocode of the algorithm is given in Algorithm~\ref{algo1}. The variables used by each robot are described in Table~\ref{table:notations}. We now describe the two stages of the algorithm\onlyLong{; Fig.~\ref{fig:algo1} illustrates the working principle of the stages}.
%At their core, each of the two stages uses a modified version of the DFS traversal described in Section~\ref{section:basic}.

\onlyLong{\subsubsection{Stage 1}}
\onlyShort{\vspace{1mm} \noindent{\bf Stage 1.}}
We first introduce some terminology. A settled/unsettled robot $i$ is one for which $i.settled = 1/0$. For brevity, we say a node is settled if it has a settled robot. 
At the start of stage 1, there may be multiple ($\geq 1$) unsettled robots at some of the nodes. 
Let $U^{s1}$/$U^{e1}$/$U^{e2}$ be the set of unsettled robots at a node at the start of stage 1/end of stage 1/end of stage 2.  In general, we define a $U$-set to be the (non-empty) set of unsettled robots at a node. Let the lowest robot ID among $U^{s1}$ at a node be $U^{s1}_{min}$. We use $r$ to denote a settled robot.

%%PSEUDOCODE-Begin%%%

%\onlyShort{\vspace{-5mm}}
\begin{algorithm}[H]
{\onlyShort{\footnotesize} 
\onlyLong{\footnotesize}
\If{$i$ is alone at node}{ 
 $i.settled\leftarrow 1$; do not set $i.treelabel$
}
\For{$pass = 1, \log\,k$}{
 \underline{Stage 1 (Graph\_DFS: for group dispersion of unsettled robots)}\\
 \For{$round = 0, \min(4m-2n+2,k\Delta)$}{
  \If{visited node is free}{
  highest ID robot $r$ settles; $r.treelabel \leftarrow$ $x$.ID, where $x$ is robot with lowest ID\\
  $x$ continues its DFS after $r$ sets its $parent$, $child$ for DFS of $x$\\
  other visitors follow $x$
  }
  \ElseIf{visited node has a settled robot $r$}{
   \If{$r.treelabel< x.ID$ for visitors $x$}{
   all visiting robots: stop until ordered to move
   }
   \ElseIf{$r.treelabel\leq y.ID$ for visitors $y$ and $r.treelabel=x.ID$ for some visitor $x$}{
   $x$ continues its DFS after $r$ updates $child$ if needed\\
   all other unsettled robots follow $x$
   }
   \ElseIf{visitor $x (x\neq r)$ has lowest ID and lower than $r.treelabel$}{
   $r.treelabel \leftarrow x.ID$\\
   $x$ continues its DFS after $r$ sets its $parent$, $child$ for DFS of $x$\\
   all other unsettled robots follow $x$
   }
  }
 }
 All settled robots: reset $parent$, $child$\\
 \underline{Stage 2 (Connected\_Component\_DFS\_Traversal: for gathering unsettled robots)}\\
 All robots: $mult\leftarrow$ count of local robots\\
 \If{$i$ has the lowest ID among unsettled robots at its node}{
 $i.home\leftarrow r.ID$, $r.treelabel\leftarrow i.ID$, where $r$ is the settled robot at that node\\
 $i$ initiates DFS traversal of connected component of nodes with settled robots 
 }
% \tcp{If visited node is free in the DFS, ignore and retrace step (go back) and continue.}
 \For{$round = 0, \min(4m-2n+2,2k\Delta)$}{
  \If{visited node is free}{
   ignore the node; all visitors backtrack, i.e., retrace their step
  }
  \ElseIf{visited node has a settled robot $r$}{
   \If{lexico-priority of $r$ is highest and greater than that of all visitors}{
   all visiting robots: stop until ordered to move
   }
   \ElseIf{lexico-priority of $r$ is highest but equal to that of some visitor $x$}{
   $x$ continues its DFS traversal after $r$ updates $child$ if needed (until $x.home=r.ID$ and all ports at the node where $r$ is settled are explored)\\
   all other unsettled robots: follow $x$ if $x.home\neq r.ID$
   }
   \ElseIf{visitor $x (x\neq r)$ has highest lexico-priority and higher than that of $r$}{
   $r.treelabel\leftarrow x.ID$, $r.mult\leftarrow x.mult$\\
   $x$ continues its DFS traversal after $r$ sets $parent$, $child$ for DFS of $x$\\
   all other unsettled robots follow $x$
   }
  } 
 }
 reset $parent$, $child$, $treelabel$, $mult$, $home$\\
\onlyLong{
\tcp{Lexico-priority: $(mult, treelabel/ID)$. Higher $mult$ is higher priority; if $mult$ is equal, lower $treelabel/ID$ has higher priority.}
\tcp{In each connected component of settled nodes, only one robot will return to its home node, collecting all unsettled robots to it.}
}
}
}
\caption{Algorithm {\em Graph\_Disperse(k)} to solve \dis.}
\label{algo1}
\end{algorithm}

%%Pseudocode-End%%%
\newpage
In stage 1, the unsettled robots at a node begin $DFS(|U^{s1}|)$, following the lowest ID $(= U^{s1}_{min})$ robot among them. Each instance of the DFS algorithm, begun concurrently by different $U^{s1}$-sets from different nodes, induces a DFS tree in which the settled nodes have robots with the same $treelabel$, which is equal to the corresponding $U^{s1}_{min}$. During this DFS traversal, the robots visit nodes, at each of which there are four possibilities. The node may be free, or may have a settled robot $r$, where $r.treelabel$ is less than, equals, or is greater than $x.ID$, where $x$ is the visiting robot with the lowest ID. The second and fourth possibilities indicate that two DFS trees, corresponding to different $treelabel$s meet. As each robot is allowed only $O(\log(\max(k,\Delta)))$ bits memory, it can track the variables for only one DFS tree. 
%\gok{cutting here}
%\onlyShort{We deal with these possibilities as shown for lines 6, 11, 13, 16 of the pseudo-code.}
%\onlyLong{
We deal with these possibilities as described below. 
\begin{enumerate}
    \item If the node is free (line 6), the logic of $DFS(k)$ described in Section~\ref{section:basic} is followed. Specifically, the highest ID robot from the visiting robots (call it $r$) settles, and sets $r.settled$ to 1 and $r.treelabel$ to $x.ID$. Robot $x$ continues its DFS, after setting $r.parent$, $r.child$ and $r.phase$ for its own DFS as per the logic of $DFS(k)$ described in Section~\ref{section:basic}; and other visiting robots follow $x$. 
    \item If $r.treelabel<x.ID$ (line 11), all visiting robots stop at this node and discontinue growing their DFS tree.
    \item If $r.treelabel=x.ID$ (line 13), robot $x$'s traversal is part of the same DFS tree as that of robot $r$. Robot $x$ continues its DFS traversal and takes along with it all unsettled (including stopped) robots from this node, after updating $r.child$ if needed as per the logic of $DFS(k)$ described in Section~\ref{section:basic}.
    \item If $r.treelabel>x.ID$ (line 16), robot $x$ continues growing its DFS tree and takes along all unsettled robots from this node with it. To continue growing its DFS tree, $x$ overwrites robot $r$'s variables set for $r$'s old DFS tree by including this node and $r$ in its own DFS tree. Specifically, $r.treelabel\leftarrow x.ID$, $r.parent$ is set to the port from which $x$ entered this node, and $r.child$ is set as per the logic described for $DFS(k)$ in Section~\ref{section:basic}.
\end{enumerate}

Note that if the robots stop at a node where $r.treelabel < x.ID$, they will start moving again if a robot $x'$ arrives such that $x'.ID\leq r.treelabel$. At the end of stage 1, either all the robots from any $U^{s1}$ are settled or some subset of them are stopped at some node where $r.treelabel<U^{s1}_{min}$.
%}

%Let $U^{s1}$/$U^{e1}$/$U^{e2}$ be the set of unsettled robots at a node at the start of stage 1/end of stage 1/end of stage 2.  In general, we define a $U$-set to be the (non-empty) set of unsettled robots at a node. Let the lowest robot ID among $U^{s1}$ at a node be $U^{s1}_{min}$. The unsettled robots at a node begin the DFS, following the lowest ID $(U^{s1}_{min}$) robot among them. If they arrive at a free node (i.e., a node without a settled robot), the highest ID robot among them settles there and the others continue the DFS. If they arrive at a node with a lower $treelabel$ than $U^{s1}_{min}$, they stop there. If they arrive at a node with a higher or equal $treelabel$ than $U^{s1}_{min}$, they continue the DFS and take along any other unsettled robots from that node \anis{how? by erasing the previous DFS tree or follow the previous tree?}. Note that if the robots stop at a node where $treelabel < U^{s1}_{min}$, they will start moving again if a robot arrives with ID $x'$ such that $x'<treelabel$. \anis{this meeting a lower treelabel or higher treelabel may be described a little bit more. How the tree is constructing and how it's getting a label---is not described here. If we describe them before, then we don't need here. Otherwise, we better include them here.}\anis{when stage 1 ends?}

\begin{lemma}\label{stageoneresult}
%Let $U^{s1}$ be the set of unsettled robots at a node at the start of stage 1 and let the lowest robot ID among them be $U^{s1}_{min}$. 
For any $U^{s1}$-set, at the end of stage 1, either (i) all the robots in $U^{s1}$ are settled or (ii) the unsettled robots among $U^{s1}$ are present all together along with robot with ID $U^{s1}_{min}$ (and possibly along with other robots outside of $U^{s1}$) at a {\em single node} with a settled robot $r$ having $r.treelabel<U^{s1}_{min}$. 
\end{lemma}
%\onlyLong{
\begin{proof}
The DFS traversal of the graph can complete in $4m-2n+2$ steps as each tree edge gets traversed twice, and each back edge, i.e., non-tree edge of the DFS tree,  gets traversed 4 times (twice in the forward direction and twice in the backward direction) if the conditions in lines (6), (13), or (16) hold. The DFS traversal of the graph required to settle $k$ robots and hence discover $k$ new nodes, can also complete in $k\Delta$ steps as a node may be visited multiple times (at most its degree which is at most $\Delta$ times). 
%As all $n$ nodes may be visited and $n\geq k \geq |U|$, possibility (i) is evident.
As $k \geq |U^{s1}|$, possibility (i) is evident.

In the DFS traversal, if condition in line (11) holds, the unsettled robots remaining in $U^{s1}$, including that with ID $U^{s1}_{min}$, stop together at a node with a settled robot $r'$ such that $r'.treelabel<U^{s1}_{min}$. They may move again together (lines (15) or (19)) if visited by a robot with ID $U'_{min}$ equal to or lower than $r'.treelabel$ (lines (13) or (16)), and may either get settled (possibility (i)), or stop (the unsettled ones together) at another node with a settled robot $r''$ such that $r''.treelabel<U'_{min}$. This may happen up to $k-1$ times. However, the remaining unsettled robots from $U^{s1}$ never get separated from each other. If the robot with ID $U^{s1}_{min}$ is settled at the end of stage 1, so are all the others in $U^{s1}$. If $U^{s1}_{min}$ robot is not settled at the end of stage 1, the remaining unsettled robots from $U^{s1}$ have always moved and stopped along with $U^{s1}_{min}$ robot. This is because, if the robot with ID $U^{s1}_{min}$ stops at a node with settled robot $r'''$ (line 12), $r'''.treelabel<U^{s1}_{min}$ and hence $r'''.treelabel$ is also less than the IDs of the remaining unsettled robots from $U^{s1}$. If the stopped robot with ID $U^{s1}_{min}$ begins to move (line 15 or 19), so do the other stopped (unsettled) robots from $U^{s1}$ because they are at the same node as the robot with ID $U^{s1}_{min}$. Hence, (ii) follows. 
\end{proof}
%}

Let us introduce some more terminology. Let ${\mathcal U}^{s1}$ be the set of all $U^{s1}$. Let ${\mathcal U}^{s1}_{min}$ be 
$\min_{U^{s1} \in {\mathcal U}^{s1}}(U^{s1}_{min})$. 
The set of robots in that $U^{s1}$ having $U^{s1}_{min}= {\mathcal U}^{s1}_{min}$ are dispersed at the end of stage 1 because the DFS traversal of the robots in that $U^{s1}$ is not stopped at any node by a settled robot having a lower $treelabel$ than that $U^{s1}_{min}$. 
Let $u^{s1}_p$, $u^{e1}_p = u^{s2}_p$, and $u^{e2}_p$ denote the number of nodes with unsettled robots at the start of stage 1, at the end of stage 1 (or at the start of stage 2), and at the end of stage 2 respectively, all for a pass $p$ of the algorithm.  Thus, $u^{s1}_p$ (= $|{\mathcal U}^{s1}_p|$) is the number of $U$-sets at the start of stage 1 of pass $p$. Analogously, for $u^{e1}_p = u^{s2}_p$, and $u^{e2}_p$.
We now have the following corollary to Lemma~\ref{stageoneresult}.

\begin{corollary}\label{calu}
%If ${\cal U}^{s1}$ and ${\cal U}^{e1}$ are ${\cal U}$ at the start and end of of Stage 1, respectively, then $|{\cal U}^{e1}| \leq |{\cal U}^{s1}| - 1$.
$u^{e1}_p \leq u^{s1}_p -1$.
\end{corollary}

In stage 1, each set of unsettled robots $U^{s1}$ induces a partial DFS tree, where the $treelabel$ of settled robots is $U^{s1}_{min}$. This identifies a sub-component $SC_{U^{s1}_{min}}$.
Note that some subset of $U^{s1}$ may be stopped at a node outside $SC_{U^{s1}_{min}}$, where the $treelabel < U^{s1}_{min}$.

\begin{definition}\label{scx}
A sub-component $SC_{\alpha}$ is the set of all settled nodes having $treelabel=\alpha$. ${\mathcal SC}$ is used to denote the set of all SCs at the end of stage 1.
\end{definition}
%Let us term the nodes in the partial DFS tree sharing the same $treelabel$ a sub-component. 

\begin{theorem}\label{onetoone}
There is a one-to-one mapping from the set of sub-components ${\mathcal SC}$ to the set of unsettled robots ${\mathcal U}^{s1}$. The mapping is given by: $SC_{\alpha} \mapsto U^{s1}$, where $\alpha=U^{s1}_{min}$.
\end{theorem}

%\onlyLong{
\begin{proof}
From Definition~\ref{scx}, each $SC_{\alpha}$ corresponds to a $treelabel=\alpha$. The $treelabel$ is set to the lowest ID among visiting robots, and this corresponds to a unique set of unsettled robots $U^{s1}$ whose minimum ID robot has ID $\alpha$, i.e., $U^{s1}_{min}=\alpha$. 
\end{proof}
%}

\begin{lemma}\label{sc}
Sub-component $SC_{\alpha}$ is a connected sub-component of settled nodes, i.e., for any $a,b \in SC_{\alpha}$, there exists a path $(a,b)$ in $G$ such that each node on the path has a settled robot. %having $treelabel=x (=U_{min})$. 
\end{lemma}

\onlyLong{
\begin{proof}
For any nodes $a$ and $b$ in $SC_{\alpha}$, the robot with ID $U_{min}$ $(=\alpha)$ has visited $a$ and $b$. Thus there is some path from $a$ to $b$ in $G$ that it has traversed. 
On that path, if there was a free node, a remaining unsettled robot from $U$ (there is at least the robot with ID $U_{min}$ that is unsettled) would have settled there. 
Thus there cannot exist a free node on that path and the lemma follows.
%At $a$, $U_{min}$ was the minimum ID robot that has visited. Along the $(a,b)$ path, the $treelabel$ of the nodes may only decrease. 
\end{proof}
}

Within a sub-component, there may be stopped robots belonging to one or more different sets $U^{s1}$ (having a higher $U^{s1}_{min}$ than the $treelabel$ at the node where they stop). There may be multiple sub-components that are adjacent in the sense that they are separated by a common edge. Together, these sub-components form a connected component of settled nodes. 

\begin{definition}\label{ccsn}
A connected component of settled nodes (CCSN) is a set of settled nodes such that for any $a,b \in CCSN$, there exists a path $(a,b)$ in $G$ with each node on the path having a settled robot.
\end{definition}

\begin{lemma}\label{scinccsn}
If not all the robots of $U^{s1}$ are settled by the end of stage 1, then $SC_{U^{s1}_{min}}$ is part of a CCSN containing nodes from at least two sub-components.
\end{lemma}

%\onlyLong{
\begin{proof}
Let the unsettled robots in $U^{s1}$ begin from node $a$.
The unsettled robots of $U^{s1}$ stopped (line 12), and possibly moved again (line 15 or 19) only to be stopped again (line 12), $c$ times, where $|{\mathcal U}^{s1}| > c\geq 1$. 

Consider the first time the robots arriving along edge $(u,v)$ were stopped at some node $v$. $U^{s1}_{min} > r.treelabel$, where robot $r$ is settled at $v$. Henceforth till the end of stage 1, $r.treelabel$ is monotonically non-increasing, i.e., it may only decrease if a visitor arrives with a lower ID (line 16).
%The node $u$ must have a settled robot, say $r'$, such that $U^{s1}_{min}\geq r'.treelabel>r.treelabel$. ($U^{s1}_{min} > r'.treelabel$ if $U^{s1}$ is joined by another group of robots $U'^{s1}$ such that $U'^{s1}_{min} = r'.treelabel$.) 
The path traced from $a$ to $u$ must have all settled nodes, each belonging to possibly more than one sub-component, i.e., possibly in addition to $SC_{U^{s1}_{min}}$, at the end of stage 1, which together form one or more adjacent sub-components. In any case, these sub-components are necessarily adjacent to the sub-component $SC_{\alpha}$, where $\alpha=r.treelabel$. Thus, at least two sub-components including $SC_{U^{s1}_{min}}$ and $SC_{\alpha}$ are (possibly transitively) adjacent and form part of a CCSN.

Extending this reasoning to each of the $c$ times the robots stopped, it follows that there are at least $c+1$ sub-components in the resulting CCSN.
(Additionally, 
(1) unsettled robots from the sub-component that stopped the unsettled robots of $U^{s1}$ for the $c$-th time may be (transitively) stopped by robots in yet other sub-components, 
(2) other groups of unsettled robots may (transitively or independently) be stopped at nodes in the above identified sub-components, 
(3) other sub-components corresponding to even lower $treelabel$s may join the already identified sub-components, 
(4) other sub-components may have a node which is adjacent to one of the nodes in an above-identified sub-component. 
This only results in more sub-components, each having distinct $treelabel$s (Definition~\ref{scx}) and corresponding to as many distinct $U$-sets (Theorem~\ref{onetoone}), being adjacent in the resulting CCSN.)
\end{proof}
%}

\begin{theorem}\label{samecc}
For any $U^{s1}$ at $a$, its unsettled robots (if any) belong to a single $U^{e1}$ at $b$, where $a$ and $b$ belong to the same connected component of settled nodes (CCSN).
\end{theorem}

\onlyLong{
\begin{proof}
From Lemma~\ref{stageoneresult}, it follows that the unsettled robots from $U^{s1}$ (at $a$) end up at a single node $b$ in the set $U^{e1}$. It follows that there must exist a path from $a$ to $b$ that these unsettled robots traversed. On this path, if there was a free node, a robot that belongs to $U^{s1}$ and $U^{e1}$ would have settled. Thus, there cannot exist such a free node. It follows that $a$ and $b$ belong to the same CCSN.
\end{proof}
}

Using the reasoning of Lemma~\ref{stageoneresult} and Corollary~\ref{calu}, if there are $s$ sub-components within a CCSN, there may be stopped (unsettled) robots at at most $s-1$ nodes. 
In stage 2, all such unsettled robots within a CCSN are collected at a single node within that component. 

\onlyLong{\subsubsection{Stage 2}}
\onlyShort{\vspace{1mm} \noindent{\bf Stage 2.}}
Stage 2 begins with each robot setting variable $mult$ %(Table \ref{table:notations}) 
to the count of robots at its node. The lowest ID unsettled robot $x$ at each node (having $mult>1$) concurrently initiates a DFS traversal of the CCSN after setting  $x.home$ to the ID of the settled robot $r$ and setting the $r.treelabel$ of the settled robot to its ID, $x.ID$. The DFS traversal is initiated by a single unsettled robot at a node rather than all unsettled robots at a node. 

In the DFS traversal of the CCSN, there are four possibilities, akin to those in stage 1. If a visited node is free (line 27), the robot ignores that node and backtracks. This is because neither the free node nor any paths via the free node need to be explored to complete a DFS traversal of the CCSN. 

If a visited node has a settled robot, the visiting robots may need to stop for two reasons. (i) Only the highest ``priority'' unsettled robot should be allowed to complete its DFS traversal while collecting all other unsettled robots. Other concurrently initiated DFS traversals for gathering unsettled robots should be stopped so that only some one traversal for gathering succeeds. (ii) With the limited memory of $O(\log n)$ at each robot, only one DFS traversal can be enabled at each settled robot $r$ in its $r.treelabel$, $r.parent$, and $r.child$. That is, the settled robot can record in its data structures, only the details for one DFS tree that is induced by one DFS traversal.
The decision to continue the DFS or stop is based, not by comparing $treelabel$ of the settled robot with the visiting robot ID, but by using a lexico-priority, defined next. 

\begin{definition}\label{lexicop}
The lexico-priority is defined by a tuple, $(mult,\,$ $ treelabel/ID)$. A higher value of $mult$ is a higher priority; if $mult$ is the same, a lower value of $treelabel$ or ID has the higher priority.
\end{definition} 

The lexico-priority of a settled robot $r$ that is visited, $(r.mult,\,$ $ r.treelabel)$, is compared with $(x.mult,x.ID)$ of the visiting robots $x$. The lexico-priority is a total order. There are three possibilities, as shown in lines (30), (32), and (35). %\onlyShort{ of the pseudo-code}. 
%\gok{cutting here}
%\onlyLong{
\begin{itemize}
\item (line 30): Lexico-priority of $r$ $>$ lexico-priority of all visitors: All visiting robots stop (until ordered later to move) because they have a lower lexico-priority than $r$. The DFS traversal of the unsettled robot $x'$ corresponding to $x'.ID=r.treelabel$ kills the DFS traversal of the visitors.
\item (line 32): The visiting robot $x$ having the highest lexico-priority among the visiting robots, and having the same lexico-priority as $r$ continues the DFS traversal because it is part of the same DFS tree as $r$. $r$ updates $r.child$ if needed as per the logic of $DFS(k)$ described in Section~\ref{section:basic}. This DFS search of $x$ continues unless $x$ is back at its home node from where it began its search and all ports at the home node have been explored. As $x$ continues its DFS traversal, it takes along with it all unsettled robots at $r$.
\item (line 35): The visiting robot $x$ having the highest lexico-priority that is also higher than that of $r$ overrides the $treelabel$ and $mult$ of $r$. It kills the DFS traversal and corresponding DFS tree that $r$ is currently storing the data structures for. Robot $x$ includes $r$ in its own DFS traversal by setting $r.treelabel \leftarrow x.ID$ and $r.parent$ to the port from which $x$ entered this node; $r.child$ is set as per the logic of $DFS(k)$ described in Section~\ref{section:basic}. Robot $x$ continues its DFS traversal and all other unsettled robots follow it.
\end{itemize}
%}

The reason we use the lexico-priority defined on the tuple rather than on just the $treelabel/ID$ is that the sub-component with the lowest $treelabel$ may have no unsettled robots, but yet some node(s) in it are adjacent to those in other sub-components, thus being part of the same CCSN. The nodes in the sub-component with the lowest $treelabel$ would then stop other traversing robots originating from other sub-components, but no robot from that sub-component would initiate the DFS traversal.

\begin{lemma}\label{complete}
Within a connected component of settled nodes (CCSN), let $x$ be the unsettled robot with the highest lexico-priority at the start of Stage 2.
\begin{enumerate}
    \item $x$ returns to its home node from where it begins the DFS traversal of the component, at the end of Stage 2.
    \item All settled nodes have the same lexico-priority as $x$ at the end of Stage 2.
\end{enumerate}
\end{lemma}

%\onlyLong{
\begin{proof}
(Part 1):
Robot $x$ encounters case in line (35) for the first visit to each node in its CCSN and includes that node in its own DFS traversal, and on subsequent visits to that node, encounters the case in line (32) and continues its DFS traversal. Within $\min(4m-2n+2,2k\Delta)$ steps, it can complete its DFS traversal of the CCSN and return to its home node. This is because it can visit all the nodes of the graph within $4m-2n+2$ steps. The robot can also visit the at most $k$ settled nodes in $2k\Delta$ steps; $k\Delta$ steps may be required in the worst case to visit the $k$ settled nodes in its CCSN and another at most $k\Delta$ steps to backtrack from adjacent visited nodes that are free. 

(Part 2): 
When $x$ visits a node with a settled robot $r$ for the first time (line 35), the lexico-priority of $r$ is changed to that of $x$ (line 36). Henceforth, if other unsettled robots $y$ visit $r$, $r$ will not change its lexico-priority (line 30) because its lexico-priority is now highest.
\end{proof}
%}

Analogous to the stage 1 execution, unsettled robots beginning from different nodes may move and then stop (on reaching a higher lexico-priority $lp$ node), and then resume movement again (when visited by a robot with lexico-priority $lp$ or higher). This may happen up to $s-1$ times, where $s$ is the number of sub-components in the CCSN. We show that, despite the concurrently initiated DFS traversals and these concurrent movements of unsettled robots, they all gather at the end of stage 2, at the home node of the unsettled robot having the highest lexico-priority (in the CCSN) at the start of stage 2.

\begin{lemma}\label{stagetworesult}
Within a connected component of settled nodes (CCSN), let $x$ be the unsettled robot with the highest lexico-priority at the start of Stage 2. All the unsettled robots in the component at the start of the stage gather at the home node of $x$ at the end of the stage.
\end{lemma}

%\onlyLong{
\begin{proof}
Let $y$ be any unsettled robot at the start of the stage. At time step $t$, let $y$ be at a node denoted by $v(t)$. Let $\tau$ be the earliest time step at which $y$ is at a node with the highest lexico-priority that it encounters in Stage 2. We have the following cases.
\begin{enumerate}
\item lexico-priority(settled robot at $v(\tau)$) $<$ lexico-priority($x$): We have a contradiction because at $t=\min(4m-2n+2,2k\Delta)$, settled robots at all nodes have lexico-priority that of $x$, which is highest.
\item lexico-priority(settled robot at $v(\tau)$) $>$ lexico-priority($x$): This contradicts the definition of $x$. 
\item lexico-priority(settled robot at $v(\tau)$) $=$ lexico-priority($x$).
\begin{enumerate}
\item $v(\tau) = x.home$: Robot $y$ will not move from $x.home$ (line 32) and the lemma stands proved.
\item $v(\tau) \neq x.home$: $y$ ends up at another node with lexico-priority that of $x$ at time step $\tau$. It will not move from node $v(\tau)$ unless robot $x$ visits $v(\tau)$ at or after $\tau$, in which case $y$ will accompany $x$ to $x.home$ and the lemma stands proved. 

We need to analyze the possibility that $x$ does not visit $v(\tau)$ at or after $\tau$. That is, the last visit by $x$ to $v(\tau)$ was before $\tau$. By definition of $\tau$, lexico-priority(settled robot at $v(\tau-1)$) $<$ lexico-priority(settled robot at $v(\tau)$) (= lexico-priority of $x$ in this case). By Lemma~\ref{complete}, $x$ is yet to visit $v(\tau-1)$, so the first visit of $x$ to $v(\tau-1)$ is after $\tau-1$. As $v(\tau-1)$ and $v(\tau)$ are neighbors and $x$ is doing a DFS, $x$ will visit $v(\tau)$ at or after $\tau+1$. This contradicts that the last visit by $x$ to $v(\tau)$ was before $\tau$ and therefore rules out the possibility that $x$ does not visit $v(\tau)$ at or after $\tau$.
\end{enumerate}
\end{enumerate}
\end{proof}
%}

\onlyLong{\subsection{Correctness}\label{correct}}
\onlyShort{\noindent{\bf Correctness.}}
Having proved the properties of stage 1 and stage 2,
%in Lemma~\ref{stageoneresult} and Lemma~\ref{stagetworesult}, 
we now prove the correctness of the algorithm.
%Let $u^{s1}_p$, $u^{e1}_p = u^{s2}_p$, and $u^{e2}_p$ denote the number of nodes with unsettled robots at the start of stage 1, at the end of stage 1 (or at the start of stage 2), and at the end of stage 2, all for round $p$ of the algorithm. We define a $U$-set to be the (non-empty) set of unsettled robots at a node. Thus, $u^{s1}_p$ is $|{\cal U}^{s1}_p|$ or the number of $U$-sets at the start of stage 1 of round $p$. Analogously, for $u^{e1}_p = u^{s2}_p$, and $u^{e2}_p$.

\begin{lemma}\label{half}
$u^{e2}_p = u^{s1}_{p+1} \leq \frac{1}{2}\cdot u^{s1}_p$
\end{lemma}

%\onlyLong{
\begin{proof}
From Lemma~\ref{stageoneresult}, for any $U^{s1}$ at the end of stage 1, (i) a set of unsettled robots $U^{s1}$ is fully dispersed, or (ii) a subset of $U^{s1}$ of unsettled robots is stopped and present together at at most one node with a settled robot $r$ such that $r.treelabel<U^{s1}_{min}$.

In case (i), there are two possibilities. 
(i.a) There is no group of unsettled robots stopped at nodes in the CCSN where the robots of $U$ have settled. In this case, this $U^{s1}$-set does not have its robots in any $U^{e1}$-set. %$u^{e1}_p$ is potentially one less than $u^{s1}_p$. 
(i.b) $z (\geq 1)$ groups of unsettled robots are stopped at nodes in the CCSN where the robots of $U$ have settled. These groups correspond to at least $z+1$ unique $U$-sets and at least $z+1$ sub-components that form a CCSN (by using reasoning similar to that in the proof of Lemma~\ref{scinccsn}). 
In case (ii), at least two sub-components, each having distinct $treelabel$s and corresponding to as many distinct $U$-sets (Theorem~\ref{onetoone}), are adjacent in the CCSN (Lemma~\ref{scinccsn}). 
%(Additionally, other groups of unsettled robots may (transitively or independently) be stopped at nodes in these sub-components, or unsettled robots from the sub-component having $treelabel$ = $s.treelabel$ may be (transitively) stopped by robots in yet other sub-components. Or other sub-components corresponding to even lower $treelabel$s may join the already identified sub-components. This only results in more sub-components, each having distinct $treelabel$s (Definition~\ref{scx}) and corresponding to as many distinct $U$-sets (Theorem~\ref{onetoone}), being adjacent in the resulting connected component of settled nodes.)

%From Lemma~\ref{stageoneresult}, we also have that any $U$-set cannot result in more than one group of unsettled robots at the end of stage 1. 
From Lemma~\ref{stageoneresult}, we also have that any $U^{s1}$-set cannot have unsettled robots in more than one $U^{e1}$.
Each robot in each $U^{s1}$-set in the CCSN, that remains unsettled at the end of stage 1, belongs to some $U^{e1}$-set that also belongs to the {\em same} CCSN (Theorem~\ref{samecc}). 
From Lemma~\ref{stagetworesult} for stage 2, all the unsettled robots in these $U^{e1}$-sets in the CCSN, are gathered at one node in that CCSN. 
Thus, each unsettled robot from each $U^{s1}$-set in the same CCSN is collected at a single node as a $U^{e2}$-set in the same CCSN. 
Thus, in cases (i-b) and (ii) above, {\em two or more} sub-components, each corresponding to a distinct $treelabel$ and a distinct $U^{s1}$-set (Theorem~\ref{onetoone}), combine into a single CCSN (Lemma~\ref{scinccsn}) and in stage 2, there is a single node with unsettled robots from all the $U^{s1}$-sets belonging to the same CCSN, i.e., a single $U^{e2}$-set, or a single $U^{s1}$-set for the next round. 
Note that each sub-component $SC_{\alpha}$ is a connected sub-component (Lemma~\ref{sc}) and hence belongs to the same CCSN; thus when sub-components merge, i.e., their corresponding $U^{s1}$-sets merge, and we have a single $U^{e2}$-set in the CCSN, there is no double-counting of the same $SC_{\alpha}$ and of its corresponding $U^{s1}$-set in different CCSNs. 
Thus, $u^{e2}_{p}$ (= $u^{s1}_{p+1}$), the number of $U$-sets after stage 2, is %less than or equal to 
$\leq \frac{1}{2}\cdot u^{s1}_{p}$, where $u^{s1}_{p}$ is the number of $U$-sets before stage 1. 
\end{proof}
%}

\begin{theorem}\label{displogk}
{\dis} is solved in $\log k$ passes in Algorithm~\ref{algo1}. %\onlyShort{algorithm $Graph\_Disperse(k)$}\onlyLong{Algorithm~\ref{algo1}}.
\end{theorem}
\begin{proof}
$u^{s1}_1 \leq k/2$. From Lemma~\ref{half}, it will take at most $\log k - 1$ passes for there to be a single $U$-set. In the first stage of the $\log k$-th pass, there will be a single $U$-set. By Lemma~\ref{stageoneresult}, case (i) holds and all robots in the $U$-set get settled. (Case (ii) will not hold because there is no node with a $treelabel$ $< U_{min}$ as all $treelabel$s of settled nodes are reset to $\top$ (the highest value) at the end of stage 2 of the previous pass and all singleton robots before the first pass settle with $treelabel=\top$ (line 2)). Thus, {\dis} will be achieved by the end of stage 1 of pass $k$.
\end{proof}

Note that the DFS traversal of stage 2 is independent of the DFS traversal of stage 1 within a pass (but the $treelabel$s are not erased), and the DFS traversal of stage 1 of the next pass is independent of the DFS traversal of stage 2 of the current pass. 

\vspace{1mm}
\noindent{\bf Proof of Theorem~\ref{theorem:0}:}
Theorem~\ref{displogk} proved that {\dis} is achieved by %\onlyShort{algorithm $Graph\_Disperse(k)$}\onlyLong{Algorithm~\ref{algo1}}. 
Algorithm~\ref{algo1}.
The time complexity is evident due to the two loops of $O(\min(4m-2n+2,2\Delta k))$ for the two stages nested within the outer loop of $O(\log k)$ passes. 
The space complexity is evident from the size of the variables: $treelabel$ ($\log k$ bits), $parent$ ($\log \Delta$ bits), $child$ ($\log \Delta$ bits), $settled$ (1 bit), $mult$ ($\log k$ bits), $home$ ($\log k$ bits), $pass$ ($\log\log k$ bits), $round$ ($O(\log n)$ bits to maintain the value $O(\min(m,k\Delta)$ for each pass) defined in Table~\ref{table:notations}. \qed

%For $\Delta\leq k$, %the memory of $O(\log (\max(k,\Delta)))$ bits becomes $O(\log k)$ bits. Moreover, 
%$\min(m,k\Delta)$ becomes $\min(m,k^2)$. 
%Therefore, % memory is optimal and 
%the time is only a $O(k\log k)$ factor away from optimal $\Omega(k)$.
\onlyLong{
Therefore, we have the following corollary to Theorem \ref{theorem:0}.%\anis{The following corollary box may be put only in the long version. The above 2 lines explanation is enough for the short version.} 
\begin{corollary}
Given $k\leq n$ robots in an $n$-node arbitrary graph $G$ with maximum degree $\Delta\leq k$, Algorithm $Graph\_Disperse(k)$ solves {\dis} in $O(\min(m,k^2)\cdot \log k)$ rounds with  $O(\log k)$ bits at each robot.
\end{corollary}
}
%For $\Delta=O(1) $, %$O(\log (\max(k,\Delta)))$ bits become $O(\log k)$ bits. Moreover, 
%$\min(m,k\Delta)$ becomes $\min(m,k)$.  %Therefore, for $\Delta=O(1)$, 
%Therefore, % memory is optimal and 
%the time is only a $O(\log k)$ factor away from optimal $\Omega(k)$.  
\onlyLong{
Therefore, we have the following corollary to Theorem \ref{theorem:0}.%\anis{The following corollary box may be put only in the long version. The above 2 lines explanation is enough for the short version.} 

\begin{corollary}
Given $k\leq n$ robots in an $n$-node arbitrary graph $G$ with maximum degree $\Delta= O(1)$, algorithm $Graph\_Disperse(k)$ solves {\dis} in $O(\min(m,k)\cdot \log k)$ rounds with  $O(\log k)$ bits at each robot.
\end{corollary}
}

\begin{comment}
We finally prove that the memory  $O(\log (\max(k,\Delta)))$ bits at each robot is in fact  optimal.

\begin{theorem}%[Lower bound]
\label{theorem:lower}
Given $k\leq n$ robots in an $n$-node arbitrary graph $G$ with maximum degree $\Delta$, an $\Omega(\log (\max(k,\Delta)))$ bits at each robot is necessary for solving {\dis}.
\end{theorem}
%\onlyLong{
\begin{proof}
The memory lower bound of $\Omega(\log k)$ bits at each robot is immediate. Consider a scenario where all $k$ robots are at a single node of $G$ in the initial configuration $C_{init}$. Since all robots run the same
deterministic algorithm, all the robots perform the  same moves. Therefore, arguing similarly as in Augustine and Moses Jr.~\cite{Augustine2018}, we have $\Omega(\log k)$ bits memory lower bound at each robot. The memory lower bound of $\Omega(\log \Delta)$ bits at each robot
comes into play when $\Delta>k$. In this situation, to correctly recognize $\Delta$ different ports, a robot needs $\Omega(\log \Delta)$ bits. Otherwise, a robot cannot move to all the neighbors of $v$, and hence {\dis} may not be achieved. %The theorem follows. % combining these two lower bounds.
\end{proof}
%}
\end{comment}

%\onlyShort{\vspace{-3mm}}
\section{Algorithm for Grid Graphs} % (Theorem \ref{theorem:1})}
\label{section:grid}
%\gok{THIS SECTION IS CURRETNLY UNDER REVISION.}
%\onlyShort{\vspace{-3mm}}
We present and analyze an algorithm, $Grid\_Disperse(k)$, that solves {\dis} for $k\leq n$ robots in $n$-node grid graphs in $O(\min(k,\sqrt{n}))$ time with $O(\log k)$ bits at each robot.
\onlyLong{$Grid\_Disperse(k)$ is the first algorithm for {\dis} in grid graphs and is optimal with respect to both memory and time for $k=\Omega(n)$. %This is the first algorithm for {\dis} in grid graphs. 
}
We first discuss algorithm $Grid\_Disperse(k),k=\Omega(n),$ for square grid graphs in Section \ref{subsection:grid-disperse-n}; $Grid\_Disperse(k), k<\Omega(n),$ is discussed in Section \ref{subsection:grid-disperse-k}. We finally describe $Grid\_Disperse(k), k\leq n$, for rectangular grid graphs \onlyShort{(deferred to full version \cite{arxiv-full} due to space constraints)}\onlyLong{in Section \ref{subsection:rectangular-grid-disperse-k}}. 

\onlyLong{
We define some terminology. 
For a grid graph $G$, the nodes on 2 boundary rows and 2 boundary columns are called {\em boundary nodes} and the 4 corner nodes  on boundary rows (or columns) are called {\em boundary corner nodes}. In an $n$-node square grid graph $G$, there are exactly $4\sqrt{n}-4$ boundary nodes.
}

\onlyLong{
\subsection{High Level Overview of the Algorithm} 
%$Grid\_Disperse(n)$}
}
%In this section, we provide a high level overview of the algorithm. 
$Grid\_Disperse(k), k=\Omega(n),$ for square grid graphs has five stages, Stage 1 to Stage 5, which execute sequentially one after another. The goal in Stage 1 is to move all the robots on $C_{init}$ to position them on the boundary nodes of $G$.  The goal in Stage 2 is to move the robots on the boundary nodes of $G$ to the four boundary corner nodes of $G$. The goal in Stage 3 is to collect all $n$ robots at one corner node of $G$.
\begin{comment}
collect 
do one of the following: 

\begin{itemize}
\item 
[(i).] For even $\sqrt{n}$, distribute equally $n$ robots to the four boundary corners of $G$ so that each corner node has exactly $n/4$ robots; 
\item [(ii).] For odd $\sqrt{n}$, collect all $n$ robots at one corner node of $G$. 
\end{itemize}
\end{comment}
The goal in Stage 4 is to distribute robots on the nodes of one  boundary row or column of $G$.
The goal in Stage 5 is to distribute the robots on a boundary row or column in Stage 4 so that each node of $G$ has exactly one robot positioned on it. 
We will show that Stages 1--5 can be performed correctly solving {\dis} in $O(\sqrt{n})$ rounds. 
%$Grid\_Disperse(n)$ immediately handles the case of $k=\Omega(n)$.
Algorithm $Grid\_Disperse(k),k<\Omega(n),$ uses the stages of $Grid\_Disperse(n)$ described above modified appropriately to handle any $k<\Omega(n)$. Particularly, $Grid\_Disperse(k)$ differentiates the cases of $\sqrt{n}\leq k<\Omega(n)$ and $k<\sqrt{n}$ and handles them through separate algorithms. We then extend all these ideas for solving {\dis} in  rectangular grid  graphs.

There are several challenges to overcome in order to execute these stages successfully in $O(\min(k,\sqrt{n}))$ rounds. The first challenge is to execute Stage 1 since robots do not have access to a consistent compass to determine which direction to follow to reach boundary nodes of $G$. % and hence the robots cannot decide on the moves the lead to the boundary nodes of $G$. 
The second challenge is to execute Stages 2-4 by moving the robots only on the boundary nodes. The third challenge is on how to move the robots in Stage 5 to disperse to all the nodes of $G$ having only one robot at each node of $G$. 
We devise techniques to overcome all these challenges which needs significant depart from the techniques based on the DFS traversal used in \cite{Augustine2018,Kshemkalyani} and described in Section \ref{section:algorithm0}. Notice that through a DFS traversal only the time bound of $O(m)=O(n)$ can be guaranteed for grid graphs (refer Table \ref{table:comparision-grid}). 

\onlyLong{
\begin{figure*}[!t] %{r}{2.4in}
%\vspace{-5mm}
\centering
\includegraphics[height=2.6in,width=5.1in]{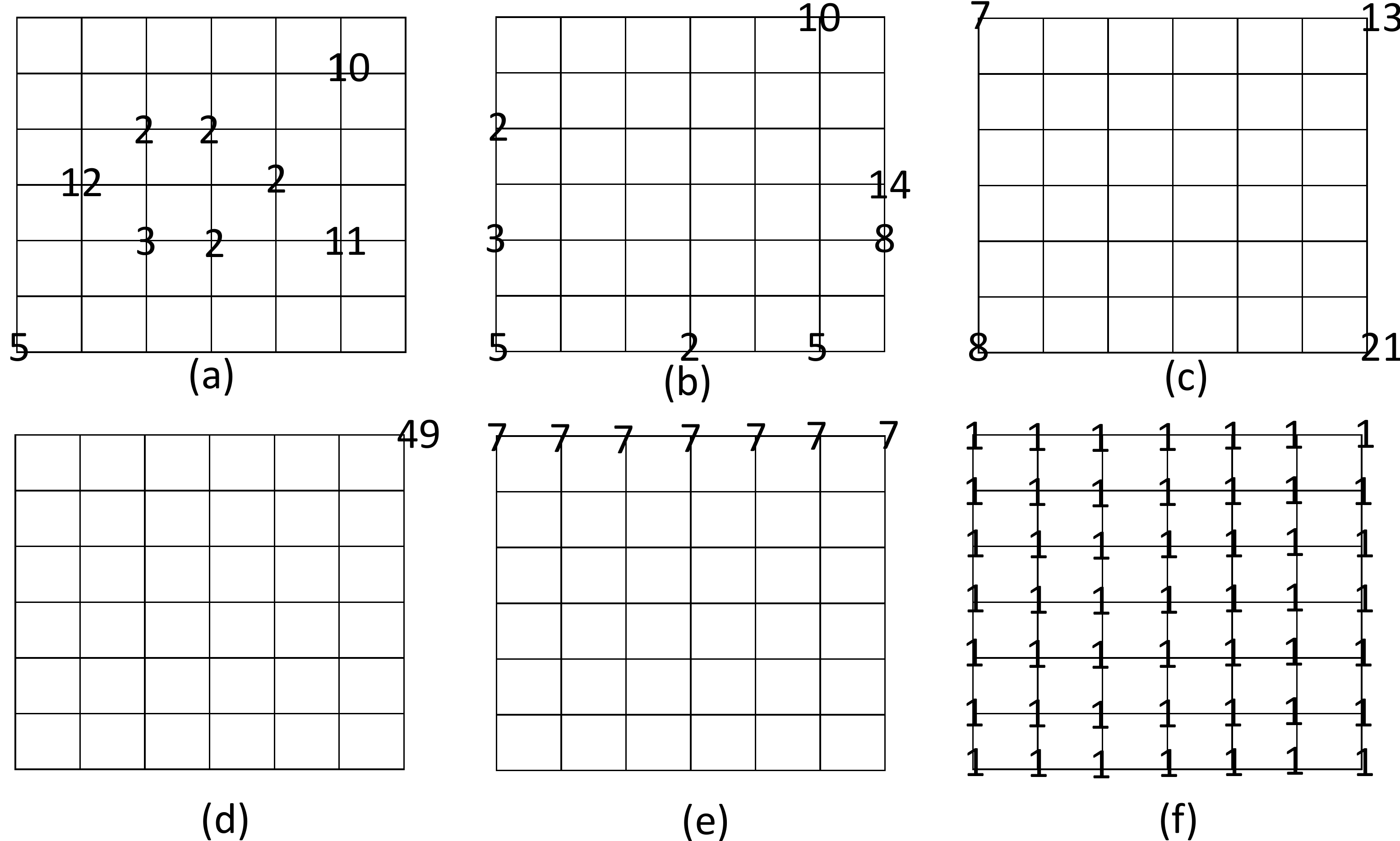}
\vspace{-2mm}
\caption{An illustration of the five stages of algorithm $Grid\_Disperse(n)$ for $n=k=49$: ({\bf a}) An initial configuration $C_{init}$, ({\bf b}) Stage 1 that moves robots to boundary nodes of $G$, ({\bf c}) Stage 2 that moves robots to four boundary corner nodes of $G$, ({\bf d}) Stage 3 that moves the robots to one boundary corner node of $G$, ({\bf e}) Stage 4 that distributes robots equally in a row (or a column) with each node having $\sqrt{n}$ robots, and ({\bf f}) Stage 5 that distributes robots so that each node of $G$ having exactly one robot each. The numbers on the grid nodes denote the number of robots positioned at that node. % during the execution of the algorithm.
}
\label{fig:algorithm4}
\vspace{-2mm}
\end{figure*}
}

\onlyShort{\vspace{-2mm}}
\subsection{Algorithm for Square Grid Graphs, $k=\Omega(n)$}
\label{subsection:grid-disperse-n}%$Grid\_Disperse(n)$}
%\onlyShort{\vspace{-2mm}}
\onlyLong{
We  describe here in detail how Stages 1--5 of $Grid\_Disperse(k),k=\Omega(n),$ are executed for square grid graphs. %This also handles the case of any $k=\Omega(n)$.  $Grid\_Disperse(k), k<\Omega(n),$ is discussed in Section \ref{subsection:grid-disperse-k}. 
Fig.~\ref{fig:algorithm4} illustrates the working principle of $Grid\_Disperse(k)$ for $n=k=49$.} %We describe the algorithm for a single robot $r\in R$. All other robots also use the same algorithm and execute their actions based on the conditions they satisfy. 
%In $C_{init}$, there may be robots on 1 or more nodes of $G$.

Each robot $r_i\in \cR$ stores five variables $r_i.round$ (initially 0), $r_i.stage$ (values 1 to 5, initially $null$), $r_i.port\_entered$ (values 0 to 3, initially $-1$), $r_i.port\_exited$ (values 0 to 3, initially $-1$), and $r_i.settled$ (values 0 and 1, initially $0$). We do not discuss how $r_i$ sets variable $r_i.round$. % in the description of $Grid\_Disperse(n)$. 
We assume that in each round $r_i$ updates its value as $r_i.round\leftarrow r_i.round+1$. 
 % and all the variables for $r_i$ are assigned $null$.
Moreover, for simplicity, we denote the rounds of each stage by $\alpha.\beta$, where $\alpha\in \{1,2,3,4,5\}$ denotes the stage and $\beta$ denotes the round within the stage. Therefore, the first round $(\alpha+1).1$ for Stage $\alpha+1$ is the next round after the last round of Stage $\alpha$.

\onlyLong{\subsubsection{Stage 1}}
%(Reposition all $N$ robots in $G$ to the boundary nodes of $G$).}

\onlyShort{\vspace{1mm} \noindent {\bf Stage 1.}} %(Gather $N$ robots at the boundary of $G$):} 
The goal in Stage 1 is to reposition $k=\Omega(n)$ robots in $C_{init}$ to at most $4\sqrt{n}-4$ boundary nodes of $G$. 
In Stage 1, a robot $r_i\in \cR$ at any node $v\in G$ moves as follows. In round 1.1, it writes $r_i.stage\leftarrow 1$ (to denote Stage 1).
If $r_i$ is already on a boundary of $G$ (i.e., $v$ is a boundary node), it does nothing in Stage 1. Otherwise, % ($r_i$ is not on the boundary node of $G$), then 
in round 1.1, $r_i$ 
picks randomly one of the four ports of $v$ and exits $v$. 
In the beginning of round 1.2, it reaches a neighbor node, say $w$, of $v$.
Let $p_w$ be the port of $w$ from which $r_i$ entered $w$. 
It assigns $p_w$ to $r_i.port\_entered$, i.e.,  $r_i.port\_entered\leftarrow p_w$.  
It then orders the three remaining ports (except $p_w$) in the clockwise order (the counterclockwise order also works) starting from $p_w$, %. %The robot $r$ can do this since each edge of $v$ are numbered from 0 to 3 ($\delta=4$).  
%then only three ports remaining at $v$. 
picks the second port  in the order starting from $p_w$, and exits $w$ using that port.  
In the beginning of round 1.3, $r_i$ reaches a neighbor node, say $u$, of $w$.
In round 1.3 and after until a boundary node is reached, $r_i$ continues similarly as in round 1.2. %IN each round, it updates the information on $r_i.port\_entered$ with the port used to enter the current node and computes the port to exit $u$ . 
\onlyLong{Variable $r_i.port\_exited$ is not used in Stage 1.}
%
%The robot $r$ does not use variables $portentered2$, $settled$, $subphase$, and $backtrack$ in Phase 1.  Robot $r$ increases the value at $round$ in every round and writes $phase\leftarrow 2$ when $round==\sqrt{N}$. This means that $r$ starts executing phase 2 when $round$ value is at least $\sqrt{N}$.
%We have the following lemma for Stage 1.

\begin{lemma}
\label{lemma:stage1}
At the end of Stage 1, all $k=\Omega(n)$ robots in $C_{init}$ are positioned on at most $4\sqrt{n}-4$ boundary nodes of $G$. Stage 1 finishes in $\sqrt{n}-1$ rounds.
\end{lemma}
\onlyLong{
\begin{proof}
Consider any robot $r_i\in \cR$. 
If $r_i$ on a boundary node of $G$ in $C_{init}$, this lemma is immediate since $r_i$ does not move in Stage 1. 
Therefore, we only need to prove that this lemma holds  for $r_i$ even when it is on a non-boundary node (say $v$) in $C_{init}$.
Let the four ports of $v$ be $p_{v1},p_{v2},p_{v3},$ and $p_{v4}$.
Suppose $r_i$ exits $v$ using $p_{v1}$ in round 1.1 and reaches node $w$ in the beginning of round 1.2. If $w$ is a boundary node, we are done.
If not,  let $L$ be the row or column of $G$ passing through nodes $w$ and $v$.
Let $L_{\overrightarrow{vw}}$ denotes one direction of $L$ starting from $v$ and going toward $w$ (the other direction be  $L_{\overrightarrow{wv}}$). 
It remains to show that in round 1.2 and after, $r_i$ always moves on the nodes on $L$ in direction $L_{\overrightarrow{vw}}$.

Let $p_{w1}$ be the port at $w$ from which $r_i$ entered $w$ in round 1.1. The three remaining ports at $w$ are $p_{w2},p_{w3},$ and $p_{w4}$. 
Since $r_i$ picks second port in the clockwise (or counterclockwise) order in round 1.2 and after, the port $r_i$ picks at $w$ is always opposite port of port $p_{w1}$ that it used to enter $w$ from $v$ in round 1.1. Therefore, in the beginning of round 1.3, $r_i$ reaches a neighbor node of $w$ on $L_{\overrightarrow{vw}}$ (opposite of $v$ on $L$). This makes $r_i$ move in the same row or column of $G$ in each subsequent move until it reaches a boundary node. 
%Since there are exactly $4\sqrt{n}-4$ nodes on the boundary of $G$ and $n$ robots, there will be 0 or more robots on those boundary nodes.  

We now prove the time bound. Since $G$ is a square, we have exactly $\sqrt{n}$ nodes in each row or column $L$. Furthermore, all the robots in $\cR$ move in every round. Therefore, $r_i$ reaches a boundary node in at most $\sqrt{n}-1$ rounds because any robot that is not already in the boundary will be at most $\sqrt{n}-1$ distance away from the boundary nodes of $G$ in $C_{init}$. %Note that the robots on the boundary of $C_{init}$ in $G$ do not move in Phase 1.
\end{proof}
}

\onlyLong{\subsubsection{Stage 2}}
%(Collect all $N$ robots on the boundary nodes of $G$ to four corners of $G$).}
\onlyShort{\noindent{\bf Stage 2.}}
The goal in Stage 2 is to collect all $k=\Omega(n)$ robots on %at most $4\sqrt{n}-4$ 
boundary nodes of $G$ to four boundary corners of $G$. 
In round 2.1, % (which is round $\sqrt{n}$), 
$r_i$ sets $r_i.stage\leftarrow 2$ (to denote Stage 2). 
Let $L_{ab}$ be a boundary row or column of $G$ passing through boundary corners $a,b$ of $G$. There are $\sqrt{n}$ boundary nodes of $G$ on $L_{ab}$. In Stage 2, $Grid\_Disperse(n)$ collects the robots on the nodes on $L_{ab}$ to node $a$ and/or $b$. 

Suppose $r_i$ is on a node $x \in L_{ab}$ in the beginning of Stage 2. 
If $x=a$ or $x=b$, it does not move in Stage 2.
If $x\neq a,b$, $r_i$ moves as follows in round 2.1.   
\begin{itemize}
    \item ({\bf Case a}) If $r_i$ did not move in Stage 1 (i.e., $r_i$ was on a boundary node in $C_{init}$), it picks randomly a port (say $p_{x1}$) among three ports $\{p_{x1},p_{x2},p_{x3}\}$ at $x$, sets $r_i.port\_exited\leftarrow p_{x1}$, and exits $x$ following $p_{x1}$. The port information written in  $r_i.port\_exited$ is used to discard the port from considering while exiting the node next time. 
    \item ({\bf Case b}) If $r_i$ moved in Stage 1 ($r_i$ was on a non-boundary node in $C_{init}$), let $p_{x1}$ be the port at $x$ from which $r_i$ entered $x$ in Stage 1 (i.e., $r_i.port\_entered \leftarrow p_{x1}$). Then, $r_i$ picks randomly a port (say $p_{x2}$) between two ports $p_{x2}$ and $p_{x3}$ and exits $x$ following $p_{x2}$.    
\end{itemize}

In the beginning of round 2.2, $r_i$ reaches a neighbor node (say $y$) of $x$. If $y=a$ or $y=b$, Stage 2 finishes for $r_i$. 
Otherwise, we have two cases:
\begin{itemize}
    \item ({\bf Case a.1}) $y$ is a node on $L_{ab}$ (i.e., a boundary node). In {\bf Case b}, $y$ is definitely on $L_{ab}$. However, in {\bf Case a}, $y$ is on  $L_{ab}$ for two ports. Let $p_{y1}$ be the port at $y$ from which $r_i$ entered $y$, i.e., $r_i.port\_entered=p_{y1}$. In round 2.2, $r_i$ picks randomly one (say $p_{y2}$) among two ports $p_{y2}$ and $p_{y3}$, sets $r_i.port\_exited\leftarrow p_{y2}$, and exits $y$ following $p_{y2}$.  In the beginning of round 2.3, $r_i$ reaches a neighbor (say $z$) of $y$. %, which is either a boundary node or a non-boundary node. 
    In round 2.3, if $z$ is a boundary node, we have a scenario similar as described above for round 2.2.
    If $z$ is not a boundary node, %not a boundary node, 
    let $p_{z1}$ be the port of $z$ from which $r_i$ entered $z$, $r_i.port\_entered=p_{z1}$. In round 2.3, $r_i$ exits $z$ following $p_{z1}$. This takes $r_i$ back to $y$ in the beginning of round 2.4. In round 2.4, $r_i$ picks only remaining port $p_{y3}$ (port $p_{y1}$ was taken while entering $y$ from $x$ in round 1.1 and port $p_{y2}$ was taken while entering $z$ in round 2.2) and exits $y$ following $p_{y3}$.  In the beginning of round 2.5, 
    $r_i$ will be on $L_{ab}$. % we have a scenario similar to {\bf Case c} for round 2.2. 
    %\end{itemize}
    \item ({\bf Case a.2}) $y$ is not a node on $L_{ab}$ (i.e., a non-boundary node). This happens in {\bf Case a}  if $r_i.port\_exited$ leads to a non-boundary node. In this case, let $p_{y1}$ be the port at $y$ from which $r_i$ entered $y$. In round 2.2, robot $r_i$ exits $y$ using port $p_{y1}$. This takes $r_i$ back to the boundary node $x$ in the beginning of round 2.3. In round 2.3, $r_i$ picks randomly one (say $p_{x2}$) between two remaining ports $p_{x2}$ and $p_{x3}$ and exits $x$ following $p_{x2}$. In the beginning of round 2.4, we have a scenario similar to {\bf Case b} in round 2.1.  
\end{itemize}

%We have the following lemma for Stage 2.  %{\bf think of the technique may be move to the second topmost row and back and continue. May be memory of two previous ports help}

\begin{lemma}
At the end of Stage 2, all $k=\Omega(n)$ robots in $G$ are positioned on (at most) $4$ boundary corner nodes of $G$. Stage 2 finishes in $3(\sqrt{n}-1)$ rounds after Stage 1. % and in total $4(\sqrt{N}-1)$ rounds. 
\end{lemma}
\onlyLong{
\begin{proof}
In the beginning of Stage 2, if $r_i$ is already on a boundary corner node, then this lemma is immediate. 
Therefore, suppose $r_i$ is not on a boundary node (say $x$) in the beginning of round 2.1.
We have two cases: {\bf (Case a)} $r_i$ was on $x$ in $C_{init}$; {\bf (Case b)} $r_i$ moved in Stage 1 to reach $x$.

We first discuss {\bf Case a} on how $r_i$ moves in round 2.1. Since $r_i$ has not moved in Stage 1, it does not have information on from what port at $x$ it entered $x$, i.e., $r_i.port\_entered =null$. Since $x$ is a boundary node, it has three ports $p_{x1},p_{x2},p_{x3}$. $r_i$ picks a port (say $p_{x1}$), sets $r_i.port\_exited\leftarrow p_{x1}$, and exits $x$. 
In {\bf Case b}, $r_i.port\_entered \neq null$. $r_i$ picks a port between two ports (except port $r_i.port\_entered$) and exits $x$.

Suppose $y$ be a node in which $r_i$ arrives in the beginning of round 2.2.  
In {\bf Case a}, $y$ may be a boundary ({\bf Case a.1}) or  non-boundary node ({\bf Case a.2}), however, in {\bf Case b}, $y$ is a boundary node.
Note that $r_i$ can figure out whether it is on a boundary or non-boundary node. 
For {\bf Case a.2}, $r_i$ exits in round 2.2 following the port used to enter $y$ in round 2.1; $r_i$ has that information in $r_i.port\_exited$ set while moving in round 2.1. 
This takes $r_i$ back to $x$ in the beginning of round 2.3.
Now in round 2.3, $r_i$ exits $x$ using one of the two remaining ports, which takes it to a boundary node $z$ in the beginning of round 2.4 (as in {\bf Case a.1} or {\bf Case b} in round 2.1). 

Therefore, round 2.2 of $r_i$ for {\bf Cases a.1} and {\bf b}  and round 2.4 of {\bf Case a.2} are the same. 
That means, $r_i$ is on a boundary node $y$ in round 2.2 for {\bf Cases a.1} and {\bf b} and in round 2.4 for {\bf Case a.2}.
In these cases, $r_i.port\_entered$ has information on a boundary port of $y$ leading to $x$.   
Therefore, now $r_i$ has a choice between one boundary port and another non-boundary port of $y$ to exit $y$ in round 2.2 or 2.4. If $r_i$ exits using a boundary port of $y$ (not $r_i.port\_entered$), $r_i$ reaches a boundary neighbor, say $z$ of $y$, and round 2.3 or 2.5 is equivalent to round 2.2 or 2.4.

If $r_i$ exits using a non-boundary port of $y$ in round 2.2 or 2.4, in round 2.3 or 2.5, it returns back to $y$. In round 2.4 or 2.6, $r_i$ has only one port remaining, which is a boundary port of $y$, to exit $y$, taking $r_i$ to a boundary neighbor node $z$ of $y$. Note that this is possible through not taking ports of $y$ that are in $r_i.port\_entered$ and $r_i.port\_exited$ variables. 

It now remains to show that $r_i$ always moves in a same direction of the boundary row or column during Stage 2. This can be easily shown similar to Stage 1 since $r_i$ always discards the ports through which it entered a boundary node from another boundary/non-boundary node by writing the port information in $r_i.port\_entered$ and $r_i.port\_exited$ variables. 

We now prove the time bound for Stage 2. We have that each row and column of $G$ has $\sqrt{n}$ nodes. Moreover, $r_i$ is at most $\sqrt{n}-1$ nodes away from a boundary corner node of $G$. While moving in Stage 2, $r_i$ reaches a neighbor node in at most 3 rounds (one round to a non-boundary node, one round to be back from the non-boundary node, and then definitely to a boundary node). Therefore, in total, $3(\sqrt{n}-1)$ rounds after Stage 1 finishes.
\end{proof}
}

\onlyLong{\subsubsection{Stage 3}}
\onlyShort{\noindent{\bf Stage 3.}}
%(Either equally distribute $N$ robots to the four corners of $G$ or collect all $N$ robots at a corner of $G$).} 
The goal in Stage 3 is to collect all $k=\Omega(n)$ robots on a boundary corner node of $G$. 
In round 3.1, $r_i$ sets $r_i.stage\leftarrow 3$ (to denote Stage 3). 

Let $a,b,c,d$ be the four boundary corner nodes of $G$. 
Suppose the smallest ID robot $r_1\in \cR$ is positioned on $a$.
If $r_i$ is already on $a$, it does nothing in Stage 3. Otherwise, it is on $b,c,$ or $d$ (say $b$) and it moves in Stage 3 to reach $a$. 
\onlyShort{Since $r_i$ needs to move on the boundary of $G$, the technique of Stage 2 can be modified for $r_i$ so that it reaches $a$ following the nodes in the boundary (details in full version \cite{arxiv-full}). We have the following lemma.}

\onlyLong{
In round 3.1, $r_i$ picks randomly one of the two ports at $b$ ($b$ is a boundary corner node in $G$) and exits $b$. In the beginning of round 3.2, $r_i$ reaches a neighbor, say $b_1$, of $b$. Notice that $b'$ is a boundary node. 
Let $p_{b'1}$ be the port at $b'$ from which $r_i$ entered $b'$. 
In round 3.2, $r_i$ picks in random between two remaining ports $p_{b'2}$ and $p_{b'3}$, sets $r_i.port\_exited$, and exits $b'$. In the beginning of round 3.3, $r_i$ reaches a neighbor, say $b''$, of $b'$.  We have two cases. 
\begin{itemize}
\item If $b''$ is a boundary node, then in round 3.3, it uses the technique similar to round 3.2 to exit $b''$.
\item
If $b''$ is a non-boundary neighbor of $b'$, in round 3.3, it uses the technique of {\bf Case a.2} (Stage 2) to return back to $b'$. In round 3.4, $r_i$ uses the technique in {\bf Case a.1} to exit $b'$. 
\end{itemize}

If $r_i$ reaches to a corner (say $c\neq a$) in Stage 3, then it uses the only port that is not used while entering $c$ and continues Stage 3. $r_i$ stops moving in Stage 3 as soon as it reaches $a$. 
}

\begin{lemma}
At the end of Stage 3, all $k=\Omega(n)$ robots in $G$ are positioned on a boundary corner nodes of $G$. Stage 3 finishes in $9(\sqrt{n}-1)$ rounds after Stage 2.
\end{lemma}

\onlyLong{
\begin{proof}
The robots on the boundary corner node where $r_1$ is positioned in the beginning of Stage 3 do not move during Stage 3. Let that corner be $a$. It is immediate that when robots of corners $b,c,d$ move in round 3.1, they reach a boundary node in the beginning of round 3.2. As in Stage 2, it is easy to see that any robot $r_i$ that started moving from any of $b,c,d$ follows the same direction as in round 3.1 in round 3.2 and after. While reaching an intermediate corner before reaching $a$, $r_i$ can exit through the port of the corner not used to enter that corner to continue traversing in the same direction.  While reaching $a$, $r_i$ knows that it has to stop there since $r_1$ is at $a$. 

We now prove the time bound. Note that the largest boundary distance from any of $b,c,d$ to $a$ is $3(\sqrt{n}-1)$. After started moving from any of $b,c,d$ in round 3.1, any robot $r_i$ reach a boundary neighbor node in the same direction in at most 3 rounds. Therefore, in total $9(\sqrt{n}-1)$ rounds after Stage 2, $r_i$ reaches $a$.  
\end{proof}
}

\onlyLong{\subsubsection{Stage 4}}
\onlyShort{\noindent{\bf Stage 4.}}
%(Distribute all $N$ robots on the boundary of $G$).} 
The goal is Stage 4 is to distribute $k=\Omega(n)$ robots (that are at a boundary corner node $a$ after Stage 3) to a boundary row or column so that there will be no more than $\sqrt{n}$ robots on each node. 
In round 4.1, $r_i$ sets $r_i.stage\leftarrow 4$ (to denote Stage 4). 

We first describe how $r_i$ moves in Stage 4 when it is the smallest ID robot. 
In round 4.1, it randomly picks one of the two ports of $a$ and exits $a$. 
In the beginning of round 4.2, $r_i$ reaches a boundary neighbor node, say $a'$, of $a$.
In round 4.2, $r_i$ waits for all other robots to reach $a'$.  
In round 4.3, $r_i$ uses the approach as in Stage 3 to move to a boundary neighbor node, say $a''$, of $a'$. If $r_i$ reaches a non-boundary node $a''$ in the beginning of round 4.4, it returns back to $a'$ in round 4.4, and in round 4.5, when $r_i$ moves, it reaches a boundary node $a''$. Robot $r_i$ continues this process until it reaches a node where there will be exactly $\sqrt{n}$ or less robots left. %another boundary corner node. 

We now describe how $r_i$ moves in Stage 4 when it is not the smallest ID robot. In round 4.1, it does not move. 
In round 4.2 and after, it does not leave $a$ if it is within $\sqrt{n}$-th largest robot among the robots at $a$. 
Otherwise,
in round 4.2, it moves following the port $r_1$ (the smallest ID robot) used to exit $a$ (it writes that information in $r_i.port\_exited$ in round 4.1 after $r_1$ picks the port to move).
In the beginning of round 4.3, $n-\sqrt{n}$ robots are at $a'$. The $\sqrt{n}$ largest ID robots stay on $a'$ and others exit $a''$ simultaneously with the smallest ID robot $r_1$ in round 4.3 (as described in the previous paragraph). In each new  boundary node, $\sqrt{n}$ largest ID robots stay and others exit.    
\onlyShort{\vspace{-2mm}}
 \begin{lemma}
At the end of Stage 4, all $k=\Omega(n)$ robots in $G$ are distributed on a boundary row or column of $G$ so that there will be exactly $\sqrt{n}$ or less robots on a node. Stage 4 finishes in $3\sqrt{n}-1$ rounds after Stage 3. % and in total $23\sqrt{N}-3$ rounds.
\end{lemma} 
\onlyLong{
\begin{proof}
In round 4.1, $r_1$ moves and others can wait at $a'$. Others keep note of the port $r_1$ used to exit $a$ in their variable $port\_exited$. The robots at $a$ know the port $r_1$ used to exit $a$. In round 4.2, all robots at $a$, except $\sqrt{n}$ largest ID robots, exit $a$ using port $port\_exited$ so that they all will be at $a'$. It is easy to see that $r_1$ can wait at $a'$ in round 4.2 since it has $r_1.stage=4$ and $a'$ is not a boundary corner node. 
In round 4.3 onwards, the robots at $a'$ can simultaneously exit $a'$ using the same port $r_1$ takes to exit $a'$. Therefore, the proof of the moving on the boundary in a row or column and in the same direction while visiting new boundary nodes follows from the proofs of Stage 2 and/or 3. Furthermore, since there are $\sqrt{n}$ nodes in a row or column and $k=\Omega(n)$ robots at a corner, leaving $\sqrt{n}$ largest ID robots in each robot distributes them to the nodes of a boundary row/column. 

For the time bound, it is easy to see that in two rounds $k-\sqrt{n}$ robots reach $a'$. The boundary neighbor node of $a'\neq a$ is reached in next three rounds. Therefore, in total, $3(\sqrt{n}-1)+2=3\sqrt{n}-1$ rounds after Stage 3, Stage 4 finishes.   
\end{proof}
}

\onlyLong{\subsubsection{Stage 5}}
\onlyShort{\noindent{\bf Stage 5.}}
%(Distribute robots on boundary of $G$ to each node of $G$).} 
The goal in Stage 5 is to distribute robots to nodes of $G$ so that there will be exactly one robot on each node. 
In round 5.1, $r_i$ sets $r_i.stage\leftarrow 5$ (to denote Stage 5).  
Let $c$ be a boundary node with $\sqrt{n}$ or less robots on it and $r_i$ is on $c$.
In round 5.1, if $r_i$ is the largest ID robot $r_{max}$ among the robots on $c$, it settles at $c$ assigning $r_i.settled\leftarrow 1$. Otherwise, 
in round 5.1, $r_i$ moves as follows. 
While executing Stage 4, $r_i$ stores the port of $c$ it used to enter  $c$ (say $r_i.port\_entered=p_{c1}$) and the port of $c$ used by the robot that left $c$ exited through (say $r_i.port\_exited=p_{c2}$). %$r_i$ has this information in $r_i.port\_entered$ and $r_i.port\_exited$. 
Robot $r_i$ then exits through port $p_{c3}$, which is not $r_i.port\_entered$ and $r_i.port\_exited$. This way $r_i$ reaches a non-boundary node $c'$. All other robots except $r_{max}$ also reach $c'$ in the beginning of round 5.2.
In round 5.2, the largest ID robot $r_{max'}$ settles at $c'$. The at most $\sqrt{n}-2$ robots exit $c'$ using the port of $c'$ selected through the port ordering technique described in Stage 1. This process continues until a single robot remains at a node $z$, which settles there. 
%We have the following lemma for Stage 5. 

\begin{lemma}
\label{lemma:stage5}
At the end of Stage 5, all $k=\Omega(n)$ robots in $G$ are distributed such that there is exactly one robot positioned on a node of $G$. Stage 5 finishes in $\sqrt{n}$ rounds after Stage 4. % and in total $24\sqrt{N}-2$ rounds.
\end{lemma} 

\onlyLong{
\begin{proof}
In round 5.1, it is easy to see that all robots except $r_{max}$ at a boundary node $c$ exit to a non-boundary node since they can discard two boundary ports through information written at $port\_entered$ and $port\_exited$ variables. 
While at a non-boundary node, it is also easy through the proof in Stage 1 that the robots exiting the node follow the subsequent nodes in a row or column that they used in previous rounds of Stage 5. Since there are at most $\sqrt{n}$ robots, $\sqrt{n}$ nodes in a row/column, and a robot stays at a new node, each node in the row/column has a robot positioned on it.
Regarding the time bound, not-yet-settled robots move in each round. 
Since there are $\sqrt{n}$ nodes, traversing all of them needs $\sqrt{n}$ rounds.
\end{proof}
}

\begin{theorem}
\label{theorem:grid-disperse-n}
$Grid\_Disperse(k),k=\Omega(n),$ solves {\dis} correctly for $k=\Omega(n)$ robots  in an $n$-node square grid graph $G$ in  $O(\sqrt{n})$ rounds with $O(\log n)$ bits at each robot. % and in total $24\sqrt{N}-2$ rounds.
\end{theorem}
%\onlyLong{
\begin{proof}
Each stage of $Grid\_Disperse(k),k=\Omega(n),$ executes sequentially one after another. Therefore, the overall correctness of $Grid\_Disperse(k)$ follows combining the correctness proofs of Lemmas \ref{lemma:stage1}--\ref{lemma:stage5}.
The time bound of $O(\sqrt{n})$ rounds also follows immediately summing up the $O(\sqrt{n})$ rounds of each stage.
Regarding memory bits, variables $port\_entered$, $port\_exited$,  $settled$, and $stage$ take $O(1)$ bits ($\Delta=4$ for grids), and $round$ takes $O(\log n)$ bits. Moreover, two or more robots at a node can be differentiated using $O(\log k)=O(\log n)$ bits, for $k=\Omega(n)$. Therefore, a robot needs in total $O(\log n)$ bits.  
\end{proof}
%}

%\onlyShort{\vspace{-3mm}}
\subsection{Algorithm for Square Grid Graphs, $k<\Omega(n)$}
\label{subsection:grid-disperse-k}
\onlyLong{We now discuss algorithm $Grid\_Disperse(k)$ that solves {\dis} for $k<\Omega(n)$ robots.} %$Grid\_Disperse(n)$ immediately works for $k=\Omega(n)$.
For $\sqrt{n}\leq k<\Omega(n)$, $Grid\_Disperse(k),k=\Omega(n),$ can be modified to achieve {\dis} in $O(\sqrt{n})$ rounds. Stages 1-3 require no changes. In Stage 4, $\sqrt{n}$ robots can be left in each new node that is visited until there will be exactly $\sqrt{n}$ or less robots left at a node. Stage 5 can again be executed without changes. The minimum ID robot settles as soon as it is a single robot on a node. \onlyShort{For $k<\sqrt{n}$, the idea is to ask $r_i$ to settle (if alone at a node) or move similar to Stage 1 of $Grid\_Disperse(k)$, $k=\Omega(n),$ until it can settle at a node. The details are in \cite{arxiv-full}. We have the following theorem.
}  

\onlyLong{
Therefore, we discuss here algorithm $Grid\_Disperse(k)$ for $k<\sqrt{n}$. %Let $C_{init}$ be any initial configuration  of $k<\sqrt{n}$ robots. 
In round 1,  if $r_i$ is a single robot on a node in $C_{init}$,  it settles at that node assigning $r_i.settled\leftarrow 1$. For the case of two or more robots on a node in $C_{init}$, in round 1, $r_i$ settles at that node if it is the largest ID robot among the robots on that node. 
If not largest, then there are two cases: ({\bf Case 1}) $r_i$ is on a non-boundary node $v$ and ({\bf Case 2}) $r_i$ is on a boundary node $v$. 
In {\bf Case 1}, $r_i$ picks randomly a port among the 4 ports and exits the node $v$ using that port. 
It then follows  the technique of Stage 1 in subsequent rounds. 
It settles as soon it reaches to a node where there is no other robot settled. 
If $r_i$ settles while reaching a boundary node, we are done. Otherwise, $r_i$ starts traversing the same row/column %it traversed to enter the boundary node 
in the opposite direction. This can be done by exiting the boundary node through the port used to enter it. Then $r_i$ follows the technique of Stage 1 until it reaches to a node when it can settle.

In {\bf Case 2}, $r_i$ picks randomly one of the 3 ports and exits $v$.  If it reaches a boundary node, it returns back to $v$ and repeats this process until it reaches a non-boundary node. This can be done through $r_i.port\_entered$ and $r_i.port\_exited$ variables.  After $r_i$ reaches a non-boundary node, it continues as in {\bf Case 1} for subsequent rounds until it settles.  

}

%We have the following theorem. 
\begin{theorem}
\label{theorem:grid-disperse-k}
Algorithm $Grid\_Disperse(k)$ solves {\dis} correctly for $k < \Omega(n)$ robots in a square grid graph $G$ in $O(\min(k,\sqrt{n}))$ rounds with $O(\log k)$ bits at each robot. % and in total $24\sqrt{N}-2$ rounds.
\end{theorem} 
\onlyLong{
\begin{proof}
For $\sqrt{n}\leq k<\Omega(n)$, the overall correctness and time bounds immediately follow from  Theorem \ref{theorem:grid-disperse-n}. For memory bound $O(\log k) =O(\log n)$ when $\sqrt{n}\leq k$, since $\log \sqrt{n}=\frac{1}{2}\log n$.  Therefore, $O(\log k)$ bits at each robot is enough.

For $k<\sqrt{n}$, when a robot $r_i$ moves from $v$ to $w$ in a row (or column) $L$ in direction $L_{\overrightarrow{vw}}$ in round 1 of $Grid\_Disperse(k)$, it is easy to proof similar to Lemma \ref{lemma:stage1} that $r_i$ moves to the nodes of $L$ in direction $L_{\overrightarrow{vw}}$ in each subsequent round. If $r_i$ settles while reaching a boundary node, we are done. If not $r_i$ returns in the opposite direction $L_{\overrightarrow{wv}}$ of $L$ starting from the boundary node. Since there are $k< \sqrt{n}$ robots and $L$ has $\sqrt{n}$ nodes, $r_i$ must settle after visiting at most $k-1<\sqrt{n}$ other nodes in $L$. 
For the time bound, $r_i$ can visit $k$ nodes of $L$ in at most $2k-1$ rounds if starting from the non-boundary node in round 1 ($k-1$ rounds to reach a boundary and $k$ rounds to reach back to a free node in the opposite direction). Starting from a boundary node, $r_i$ visits all those nodes of $L$ in $k+4$ rounds (two rounds each to go to boundary nodes and come back and then to $k$ nodes of $L$ in $k$ rounds). Regarding memory, variables $port\_entered$, $port\_exited$,  $settled$, and $stage$ take $O(1)$ bits ($\Delta=4$ for grids), and $round$ takes $O(\log k)$ bits. Moreover, two or more robots at a node can be differentiated using $O(\log k)$ bits. Therefore, in total $O(\log k)$ bits at each robot is enough for {\dis}.  

The theorem follows combining the time and memory bounds of $Grid\_Disperse(k)$ for $\sqrt{n}\leq k<\Omega(n)$ and $k<\sqrt{n}$. 
\end{proof}
}

\noindent{\bf Proof of Theorem~\ref{theorem:1}:} Theorems \ref{theorem:grid-disperse-n} and \ref{theorem:grid-disperse-k} together prove Theorem 
\ref{theorem:1} for any $k\leq n$. \qed

\onlyLong{
\subsection{Algorithm for  Rectangular Grid Graphs, $k\leq n$}
\label{subsection:rectangular-grid-disperse-k}
Algorithm $Grid\_Disperse(k),k\leq n$, can be easily extended to a $n (=x\times y)$-node rectangular grid with either $x>\sqrt{n}$ or $y>\sqrt{n}$. %greater than $\sqrt{n}$. 
Suppose the values of $x$ and $y$ are known to robots.
For $Grid\_Disperse(k), k=n$,
Stages 1--3 and 5 can be executed without any change. Stage 4 can be executed in two passes. In the first pass, $x_1=\frac{n}{\max(x,y)}$ largest ID robots can be left at each node on a row/column $L$. If $L$ is of length $y'=\max(x,y)$, then there will be exactly $x'=\min(x,y)$ robots on each node and Stage 4 finishes in one pass. If not, in the second pass, the remaining robots can traverse $L$ in the opposite direction leaving $x_2=\frac{n}{\min(x,y)}-x_1$ additional robots at each node. This way each node in $L$ has exactly $y'=x_1+x_2=x_1+\frac{n}{\min(x,y)}-x_1$ robots on it and Stage 4 finishes after the second pass. % passes. %Stage 5 then can be run with no changes. 

For $\max(x,y)\leq k<n$, $Grid\_Disperse(k),k=n,$ can be modified as follows. Stages 1-3 and 5 require no changes. In Stage 4, in the first pass,
$x_1=\lfloor\frac{k}{\max(x,y)}\rfloor$ robots can be left on each node in $L$. 
If $L$ is of length $\max(x,y)$ 
and $(k\mod \max(x,y))=0$, then there will be exactly $\frac{k}{\max(x,y)}$ robots on the nodes visited in the first pass. Stage 4 then finishes. 
If the first pass visits all $\max(x,y)$ nodes of $L$ and still some robots left (that means $(k\mod \max(x,y))\neq 0$), then in the second pass, the remaining $(k-x_1\cdot \max(x,y))< \max(x,y)$ robots visit $L$ in the opposite direction leaving 1 robot in each node of $L$ visited in this pass. 
If $L$ is of length $\min(x,y)$, then in the second pass, $x_2=\max(x,y)-x_1$ additional robots can be left at each node visited. This way, there will be between $\max(x,y)$ and  $x_1(=\lfloor\frac{k}{\max(x,y)}\rfloor)$ robots (inclusive) on each node in $L$. 

For $k<\max(x,y)$, each robot $r_i$ moves as in $Grid\_Disperse(k)$ for $k<\sqrt{n}$ (Section \ref{subsection:grid-disperse-k}). If $r_i$ cannot settle after visiting a row/column $L$ two times, it starts visiting nodes on a row/column $L'$ that is perpendicular to $L$. Robot $r_i$ visits the nodes on $L'$ as in Section \ref{subsection:grid-disperse-k} until it is settled. 

\begin{theorem}\label{thm:rectangular-grid}
Algorithm $Grid\_Disperse(k)$ solves {\dis} correctly for $k \leq n$ robots in a rectangular grid graph $G$ with $n=x\times y$ nodes in $O(\min(k,\max(x,y)))$ rounds with $\Theta(\log k)$ bits at each robot.
The runtime is optimal when $k=\Omega(n)$.
\end{theorem}
\begin{proof}
The correctness bound is immediate extending the correctness proof of Theorem \ref{theorem:grid-disperse-k}. The time bound of $Grid\_Disperse(n)$, similarly as in Theorem \ref{theorem:grid-disperse-n}, would be $O(\max(x, y))$.
For $Grid\_Disperse(k)$, $\max(x,y)\leq k<\Omega(n)$, the time  bound would be as in $Grid\_Disperse(n)$, which is $O(\max(x, y))$.  For $Grid\_Disperse(k)$ with $k< \max(x,y)$, the time would be $O(k)$ as in Theorem \ref{theorem:grid-disperse-k}. 
Therefore, the time bound for any $k\leq n$ is $O(\min(k,\max(x,y)))$ rounds, which is optimal when $k=\Omega(n)$ since there is a time lower bound of $\Omega(D)$ in any graph, and for a rectangular grid of $n=x\times y$ nodes, $D=\Omega(\max(x,y))$. 
For all cases of $k\leq n$, the memory bound would be $O(\log k)$ bits as in Theorems \ref{theorem:grid-disperse-n} and \ref{theorem:grid-disperse-k}, which is clearly optimal.  
\end{proof}
}

\section{Concluding Remarks}
\label{section:conclusion}

We have presented two results for solving {\dis} of $k\leq n$ robots on $n$-node graphs. The first result is for arbitrary graphs and the second result is for grid graphs. 
Our result on arbitrary graphs exponentially improves the $O(mk)$ runtime of the best previously known algorithm \cite{Kshemkalyani} to $O(\min(m,k\Delta)\cdot \log k)$. % using the same memory of $O(\log(\max(k,\Delta)))$ bits at each robot. 
Our result on grid graphs provides the first simultaneously memory and time optimal solution for {\dis} for $k=\Omega(n)$. Moreover, our algorithm is the first algorithm for solving {\dis} in grid graphs.

For future work, it will be interesting to solve {\dis} on arbitrary graphs with time $O(k)$ or improve the existing time lower bound of $\Omega(k)$ to $\Omega(\min(m,k\Delta))$. Another interesting direction is to remove the $\log k$ factor from the time bound in Theorem \ref{theorem:0}. 
Furthermore, it will be interesting to achieve Theorem \ref{theorem:0} without each robot knowing parameters $m, \Delta,$ and $k$.  
For grid graphs, it will be interesting to either prove an $\Omega(k)$ time lower bound or provide a $O(\sqrt{k})$ runtime algorithm for $k<\Omega(n)$.  %The third interesting direction will be to consider faulty robots; our algorithms as well as previous algorithms \cite{Augustine2018,Kshemkalyani} assume fault-free robots. The fourth interesting direction will be to solve {\dis} in dynamic graphs; so far only static graph cases are studied. \anis{removed the above lines from the short paper}
Another interesting direction will be to extend our algorithms to solve {\dis} in semi-synchronous and asynchronous settings.  %The forth direction is to study uniform scatter in    

\bibliographystyle{plain}
\bibliography{references}

\end{document}